\documentclass[a4paper,UKenglish,cleveref, autoref, thm-restate]{lipics-v2021} 
\nolinenumbers
\hideLIPIcs

\usepackage{amsfonts}
\usepackage{mathtools}
\usepackage{calc}
\usepackage{hyperref}
\usepackage{graphicx}
\usepackage{multicol}
\usepackage{amsmath}
\usepackage{stmaryrd}
\usepackage{wrapfig}
\usepackage{xspace}
\usepackage{array}
\usepackage{cleveref}
\usepackage{todonotes}

\usepackage[ruled,algo2e]{algorithm2e}

\SetKwComment{Comment}{/* }{ */}
\Crefname{algocf}{Algorithm}{Algorithms}

\usepackage{tikz}
\usepackage{pgfplots}
\usetikzlibrary{decorations.markings}
\usetikzlibrary{backgrounds,folding}
\usetikzlibrary{shapes.geometric, shapes.misc}
\pgfdeclarelayer{edgelayer}
\pgfdeclarelayer{nodelayer}
\pgfsetlayers{background,edgelayer,nodelayer,main}
\tikzstyle{every picture}=[baseline=-0.25em]
\tikzstyle{none}=[inner sep=0mm]
\tikzstyle{black dot}=[inner sep=1pt,minimum width=0pt,minimum height=0pt,fill=black,draw=black,shape=circle, font=\scriptsize]
\tikzstyle{dot}=[black dot]
\tikzstyle{white dot}=[dot,fill=white]
\tikzstyle{box}=[rectangle,fill=white,draw=black, font=\scriptsize, inner sep=2pt]
\tikzstyle{box-no-outline}=[rectangle, draw=white, fill=white, inner sep=2pt]
\tikzstyle{sdash}=[-, dashed, dash pattern=on 2pt off 1pt, draw=black]
\tikzstyle{tensor}=[isosceles triangle, isosceles triangle apex angle=90,draw,scale=0.4,shape border rotate=-90]

\pgfdeclarelayer{edgelayer}
\pgfdeclarelayer{nodelayer}
\pgfsetlayers{background,edgelayer,nodelayer,main}
\tikzstyle{every loop}=[]

\usepackage{etoolbox}

\newtoggle{extern}
\toggletrue{extern}
\togglefalse{extern}

\newcommand{\tikzfig}[1]{
		\input{./figures/#1.tikz}
}

\def\fig{}

\iftoggle{extern}{
	\usetikzlibrary{external}
	\tikzexternalize[prefix=tikzfigs/]
	
	\renewcommand{\tikzfig}[1]{
		\tikzsetnextfilename{#1}
		
		\input{./figures/#1.tikz}
}
	
}{
	
}

\newcommand{\eq}[2][~~]{
	#1
	\underset{\substack{#2}}{=}
	#1
}

\newcommand{\braket}[2]{\ensuremath{\left\langle #1 | #2 \right\rangle}}
\newcommand{\cat}[1]{\mathbf{#1}}

\allowdisplaybreaks
\title{A Unique Normal Form for Tensor Trains over Arbitrary Fields}

\author{Renaud Vilmart}{Université Paris-Saclay, Inria, CNRS, ENS Paris-Saclay, Laboratoire Méthodes Formelles, 91190, Gif-sur-Yvette, France.}{renaud.vilmart@inria.fr}{https://orcid.org/0000-0002-8828-4671}{}

\authorrunning{R.~Vilmart}



\begin{document}
	
	\maketitle
	
	\begin{abstract}
		Tensor trains (or Matrix-Product States) are a data structure used in many fields of computer science and physics. They were recently shown to generalise binary decision diagrams when used over the 2-element Galois field, prompting the question of their reducibility in such a context, when the standard approach, over real or complex number, is not amenable to finite fields.
		
		We provide here a unique normal form and associated polynomial-time reduction strategy for tensor trains over arbitrary fields. We also show how to directly extract a normal form out of a full tensor, how to get the leading index and value of a normal form, and an upper bound on the size of a fully-reduced tensor train relative to a naive storage of the full tensor.
		
		On the one hand, this work strengthens the use of tensor trains as a relevant formal tool. On the other hand, from the perspective of tensor networks, it extends the formalism to more general settings than the well-studied real and complex fields, and crucially provides the first tensor train form with the uniqueness property.
	\end{abstract}

	\section{Introduction}
	
	Tensors are multidimensional extensions of vectors and matrices \cite{Dupuy2025TTln}, used pervasively both for modelling physics (continuum mechanics, electromagnetism, general relativity, quantum mechanics, ...), and as a tool for computer science (computer vision, linear system solving, machine learning, quantum computing, ...).
	
	A tensor is a rather big object, the amount of data required to store it depending exponentially in its \emph{order} (i.e.~the number of dimensions; matrices and vectors being respectively order-2 and order-1 tensors). To deal with this issue, schemes have been devised to compress the data, by decomposing the tensor into a \emph{tensor network}: a collection of smaller-ordered tensors composed together following a given topology \cite{Orus2019}. 
	It is then possible to exploit the Singular Value Decomposition (SVD) to decrease the dimensions between the tensors, either exactly (by exploiting the rank-revealing property of the SVD), or approximately, by chopping off singular values below a given threshold, and hence performing low-rank approximations \cite{Oseledets2011Decomposition}.
	
	A first caveat of the SVD is that it is in general not unique, which means that SVD-based reduction schemes never enjoy a full uniqueness property \cite{Dupuy2025TTln}. The second is that it is only applicable to matrices over $\mathbb R$ or $\mathbb C$ (hence using floats in practice). For most considered applications, like the aforementioned, this is fine. However, it recently was shown that tensor trains (tensor networks consisting of order-3 tensors connected in a line) over $\mathbb F_2$ could be used to generalise Binary Decision Diagrams (BDDs) \cite{Onaka2025Tensor,quist2026}. In particular, it has been shown that there exist families of tensor trains over $\mathbb F_2$ of polynomial size, that would require BDDs of exponential size to represent them \cite{Onaka2025Tensor}.
	
	Binary decision diagrams \cite{Lee1959Representation} are a data structure used to represent boolean functions, and that make use of the functions' structure to reduce their size. These diagrams and their numerous variants (Zero-suppressed DDs \cite{Minato1993Zero}, BDDs with complemented edges \cite{Brace1991Efficient}, ...) can be and have been used for a plethora of applications, such as model checking \cite{Clarke1993SymbolicMC}, model counting \cite{Bryant1986GraphBased}, circuit synthesis \cite{Podlaski2016Reversible}, regular expression matching \cite{Yang2010Improving}, multiobjective discrete optimization~\cite{Bergman2016Multiobjective},~..., or simply for the storing of boolean polynomials~\cite{Brickenstein2009PolyBoRi}.
	
	The fact that BDDs are special cases of tensor trains over $\mathbb F_2$ gives an interesting bridge between two rather distant fields. In particular, it opens the possibility to use tensor trains rather than BDDs for many of the latter's application, with potentially better outcomes due to the better compactness of tensor trains. However, to properly assess the advantage obtained by using tensor trains, one would first need to address the two aforementioned caveats of SVD-based tensor train reductions: application in $\mathbb F_2$, and uniqueness.
	
	\emph{We address this in this paper, by providing a unique normal form for tensor trains over arbitrary fields, that minimises their size, and can be reached in polynomial time in their original size. We also provide an algorithm that directly builds a tensor train in normal form out of any arbitrary full tensor.}
	
	The developments in this paper are based on a matrix decomposition that is different from the SVD: the LU decomposition, whereby a matrix is decomposed as the product of a lower triangular matrix and an upper triangular matrix. Although a matrix may not have an LU decomposition, this can be fixed by first permuting its rows, leading to the so-called PLU decomposition. While this one always exists, it is not unique in general, which is inconvenient for our purposes. We hence use a variation of the LU decomposition, the LDPU decomposition, that we show to be unique. Crucially, just like the SVD, the LDPU decomposition is rank-revealing, the central feature that allows us to reduce the size of the tensor trains. While the LDPU (and subsequent tensor train normal form and reduction strategy) theoretically work with any field, in practice, it may suffer from exponential bit complexity and numerical instability, hence it is best fitted for finite fields, where such concerns disappear.
	
	The paper is structured as follows. 
	In \Cref{sec:TT}, we give some necessary background on tensor trains, a short explanation of the link between these and decision diagrams, and a string-diagram based representation to simplify their manipulation. In \Cref{sec:LDPU}, we define the LDPU decomposition, show its uniqueness, and give a few useful lemmas about matrices in row-echelon form. In \Cref{sec:NF}, we define the normal form, as well as an algorithm that directly turns a full tensor into its normal form as a tensor train. The existence and uniqueness of the tensor train is then shown there. Finally, in \Cref{sec:reduc}, we define the poly-time reduction strategy that brings any tensor train into the unique normal form defined previously. We also show how the first part of the reduction is already enough to infer some useful properties about the tensor train, namely, whether it is null or not (which induces an algorithm to check equivalence between two tensor trains), or what are its leading index and leading coefficient. All proofs are given in full in the appendix.
	
	\section{Tensor Trains}
	\label{sec:TT}
	
	Throughout this paper, we will work with tensors over arbitrary fields. Let's fix one, that we denote $F$. We write $0$ and $1$ for its additive and multiplicative neutral elements. For any tensor $T\in F^{n_1\times ... \times n_d}$, we write $T[i_1]$ for the tensor $F^{n_2\times ... \times n_d}$ obtained by setting the first index of $T$ to $i_1$. Hence, for a matrix $A$, $A[i]$ corresponds to the row $i$ of $A$, and for a vector $v$, $v[i]$ corresponds to the $i$-th element of $v$. We write $T[i_1,...,i_k]$ for $T[i_1][i_2]...[i_k]$.
	
	\begin{definition}
		Let $T\in F^{n_1\times ... \times n_d}$ be a tensor. A \emph{Tensor Train (TT) decomposition} of $T$ is a tuple of tensors $(C_1,...,C_d)$ such that:
		\begin{itemize}
			\item $C_j\in F^{n_j\times r_{j-1}\times r_j}$ ($r_0 = r_d = 1$), so that $C_j[i_j] \in F^{r_{j-1}\times r_j}$ for $0\leq i_j< n_j$
			\item $T[i_1,...,i_d] = C_1[i_1]...C_d[i_d]$
		\end{itemize}
		$d$ is called the \emph{order} of the tensor (the number of \emph{modes}), $n_j$ is called the $j$-th \emph{mode dimension}, and $C_j$ is called the $j$-th \emph{core} of the tensor train. The values $r_0,...,r_d$ are called the \emph{TT-ranks} of the tensor train. The overall number of $F$-entries, defined as follows, is called the \emph{size} of the tensor train: $s:=\sum\limits_{j=1}^{d}r_{j-1} r_j  n_j$.
		We denote $n:= \max_j(n_j)$ the maximum mode dimension, and $r:=\max_j(r_j)$ the maximum rank.
	\end{definition}
	
	The notion of tensor trains is sometimes generalised to the case where $r_0=r_d$ can take any value, which generalises the interpretation and gives more leeway in the manipulation. However, the more constrained version $r_0=r_d=1$ makes the reduction of tensor trains much easier, in particular because reaching the minimum in TT-ranks can be made one TT-rank at a time (the TT-ranks reach a minimum iff no individual TT-rank can be reduced).
	
	To make the presentation simpler in the following, we allow modes of tensors to be of dimension $0$. This also means we allow matrices with $0$ rows or with $0$ columns. Since such tensors have no entries, they are, by convention, considered null.
	
	Tensor trains allow in certain cases a more compact representation of the full tensor. That representation is polynomial in $d$, $n$ and $r$, and so are many operations on the tensor trains~\cite[Chapter~13]{Hackbusch2019Tensor}. For instance, building a tensor train that represents a linear combination of tensor trains, the element-wise product (a.k.a~Schur product) or the Kronecker product of tensor trains, can be done in polynomial time (and space). Many operations that we may want to apply on a tensor may also be implemented as a poly-time operation on the tensor train. There are hence many occurrences where the full tensor is never computed: the computation starts with a standard tensor of which we know an efficient tensor train representation, and then all operations are directly applied on the tensor train. It is then paramount to have efficient reduction strategies to keep the size of the tensor train under control.
	
	\subsection{BDDs as special cases}
	
	As a main motivation for the development of tensor trains for arbitrary fields, we briefly explain how Binary Decision Diagrams (BDDs) are special cases of tensor trains over $\mathbb F_2$, the $2$-element Galois field; where the cores are $1$-row-sparse \cite{Onaka2025Tensor}. A BDD (without variable skipping\footnote{When both the $0$ and $1$ outgoing edges of a given node $u$ point to the same node $v$, it is customary to skip node $u$ and redirect all its incoming edges to $v$. This variable skipping can easily be undone, and can only incur a linear overhead.}) is a rooted directed acyclic graph, where the nodes can be partitioned into $h$ groups such that edges can only go from one group to the next. The last group consists of 2 nodes labelled $\bot$ and $\top$. For each node except the last two, there are exactly two outgoing edges, labelled $0$ and $1$ respectively. Each group of nodes (except the last) is assigned a distinct variable, which allows the structure to represent a boolean function over these variables: one can easily evaluate the function for a given variable assignment, by following the path where the edge labels correspond to the values assigned to the variables. The evaluation of the function is then $1$ if and only if the path ends up in the node labelled~$\top$.
	
	Such a BDD is easy to translate into a tensor train $T\in \mathbb F_2^{2\times...\times 2}$. Indeed, consider the (ordered) groups $i$ and $i+1$. These nodes together with the $0$-edges between them form a bipartite graph, which we can represent by an $\mathbb F_2$ biadjacency matrix (with the rows corresponding to the nodes of group $i$, and the columns those of group $i+1$). For each group $i$ (but the last), we can then easily build a core $C_i$ such that $C_i[0]$ is the biadjacency matrix for the $0$-edges from group $i$ to group $i+1$, and $C_i[1]$ is the biadjacency matrix of the $1$-edges from group $i$ to group $i+1$. For the penultimate group, we only look at the edges to the $\top$ node (i.e.~we treat the last group as if it only contained the $\top$ node). It has been proven that this construction yields a tensor train that encodes the same boolean function. An example of a tensor train derived from a BDD is given in \Cref{fig:BDD-to-TT}.
	
	\begin{figure}[!ht]
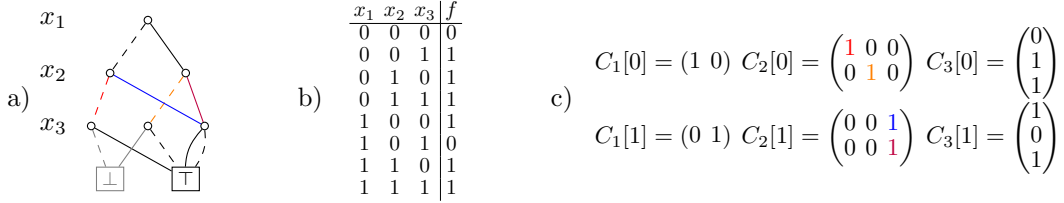

		a) $\tikzfig{BDD-example}$
		\hfill
		b) 
		\renewcommand{\arraystretch}{0.8}
		\setlength{\arraycolsep}{2pt}
		\scalebox{0.8}{$
			\begin{array}{ccc|c}
				x_1&x_2&x_3&f\\
				\hline
				0&0&0&0\\
				0&0&1&1\\
				0&1&0&1\\
				0&1&1&1\\
				1&0&0&1\\
				1&0&1&0\\
				1&1&0&1\\
				1&1&1&1
			\end{array}$}
		\hfill
		c) \scalebox{0.9}{
			$\begin{array}{ccc}
				C_1[0] = \begin{pmatrix}1&0\end{pmatrix}&
				C_2[0] = \begin{pmatrix}\color{red}1&0&0\\0&\color{orange}1&0\end{pmatrix}&
				C_3[0] = \begin{pmatrix}0\\1\\1\end{pmatrix}\\
				C_1[1] = \begin{pmatrix}0&1\end{pmatrix}&
				C_2[1] = \begin{pmatrix}0&0&\color{blue}1\\0&0&\color{purple}1\end{pmatrix}&
				C_3[1] = \begin{pmatrix}1\\0\\1\end{pmatrix}\\
			\end{array}$}
		\caption{a) A BDD on 3 variables. b) For information, the truth table of the corresponding boolean function. c) The corresponding tensor train. for ease of understanding, we have grayed out the links to $\bot$ which are ignored, and coloured every edge between the $x_2$ nodes and the $x_3$ nodes, each corresponding to a $1$ in the core $C_2$. }
		\label{fig:BDD-to-TT}
	\end{figure}
	
	It is clear now that all matrices in the cores obtained directly from a BDD are $1$-row-sparse, for each node has at most $1$ outgoing $0$-edge, and at most $1$ outgoing $1$-edge. While TTs can be exponentially more succinct than BDDs, BDDs can at best be polynomially more succinct than TTs, thanks to the BDD to TT translation that we just described \cite{Onaka2025Tensor}.
	
	By extension, Weighted Binary Decision Diagrams (WBDDs) \cite{Miller2006QMDD,Hong2022TDD} are sparse special cases of tensor trains over the corresponding field. Notice that just like for variable ordering in BDDs~\cite{Bollig1996Improving,Miller2007Sifting}, core ordering in tensor trains may dramatically change the size of the data structure~\cite{Dupuy2025TTln,Tichavsky2025Order}.
	
	Since arbitrary matrices can be seen as biadjacency matrices for associated bipartite graph, we can still have a DAG-based graphical interpretation of arbitrary tensor trains. In this case, every non-terminal node can have arbitrarily many outgoing $0$-edges et $1$-edges, and each node can be seen as encoding a vector (or tensor), built by concatenation of linear combinations of the vectors encoded by the nodes in the next group (the children nodes), starting at the $\top$ node interpreted as $\begin{pmatrix}1\end{pmatrix}$. While in BDDs (resp.~WBDDs), nodes can be merged when they represent equal (resp.~colinear) vectors, the reduction defined below can be interpreted as: whenever a node represents a vector that is a linear combination of the others, it can be removed (and connections split amongst the others).
	
	\subsection{Graphical representation}
	
	Tensor trains, also called Matrix Product States (MPS) \cite{Perez2006mps}, are a special kind of a more general data structure called tensor network \cite{Penrose1971Applications,Bridgeman2017}. As such, they can be conveniently represented graphically. To keep things completely formal, we will work with string diagrams representing the strict compact-close symmetric monoidal category $\cat{Mat}_F$ whose objects are tuples of natural numbers and morphisms $T:F^{n_1}\otimes...\otimes F^{n_k}\to F^{m_1}\otimes...\otimes F^{m_\ell}$ are linear maps from $F^{n_1}\otimes...\otimes F^{n_k}$ to $F^{m_1}\otimes...\otimes F^{m_\ell}$ \cite{kissinger2012pictures,MacLane1978}.
	
	The tensor product of two vector spaces $F^n\otimes F^m$ is another vector space, that we can define as follows: Each space $F^{n}$ has a canonical orthonormal basis $\mathcal B_n = \{e_i\}_{0\leq i < n}$ where $e_i$ is the length-$n$ vector with entry $1$ at index $i$ and $0$ everywhere else. In that case, $F^n\otimes F^m$ is the vector space whose canonical basis is $\{e_i\otimes e_j\}_{\substack{0\leq i< n\\0\leq j < m}}$. This generalises to multiple tensors of the base vector spaces $F^n$. Since the category is considered strict, we consider the associativity of $(-\otimes-)$ as an equality. If $T:F^{n_1}\otimes...\otimes F^{n_k}\to F^{m_1}\otimes...\otimes F^{m_\ell}$ and $S:F^{p_1}\otimes...\otimes F^{p_i}\to F^{q_1}\otimes...\otimes F^{q_j}$ are linear maps, $T\otimes S$ is a linear map, defined on the basis elements as:
	\begin{align*}
		T\otimes S : F^{n_1}\otimes...\otimes F^{n_k}\otimes F^{p_1}\otimes...\otimes F^{p_i} & \to F^{m_1}\otimes...\otimes F^{m_\ell}\otimes F^{q_1}\otimes...\otimes F^{q_j}\\
		e_{r_1}\otimes ... \otimes e_{r_k}\otimes e_{s_1}\otimes ... \otimes e_{s_i} &\mapsto T(e_{r_1}\otimes ... \otimes e_{r_k})\otimes S(e_{s_1}\otimes ... \otimes e_{s_i})
	\end{align*}
	The tensor product interacts with the categorical composition $(-\circ-)$ according to the bifunctorial law, as follows: if $S_i$ and $T_i$ can be composed as $S_i\circ T_i$, for $i\in\{0,1\}$, then $(S_0\otimes S_1)\circ(T_0\otimes T_1) = (S_0\circ T_0)\otimes (S_1\circ T_1)$.
	
	There is an isomorphism $\iota_{n,m}:F^n\otimes F^m\to F^{nm}$, defined on the canonical basis as $e_i\otimes e_j\mapsto e_{mi+j}$. The isomorphism $\iota_{n,m}$ is crucial to tensor manipulation, it allows us to \emph{reshape} tensors to our will. One can easily check that $\iota_{nm,p}\circ(\iota_{n,m}\otimes I_p) = \iota_{n,mp}\circ(I_n\otimes\iota_{m,p})$, which allows us to unambiguously define a generalised isomorphism $\iota_{n_1,...,n_k}:F^{n_1}\otimes...\otimes F^{n_k}\to F^{n_1...n_k}$.
	
	The space $F^1=F$ deserves some additional attention: it is the tensor unit ($F\otimes F^n\cong F^n\cong F^n\otimes F$), a property captured by $\iota_{1,n}$ and $\iota_{n,1}$. This allows us to see a size-$n$ column vector as an $F\to F^n$ morphism, and a size-$n$ row vector as an $F^n\to F$ morphism.
	
	The other special morphisms in the category are the ones that make it compact-close and symmetric:
	\begin{itemize}
		\item the symmetry $\begin{aligned}[t]
			\sigma_{n,m}: F^n\otimes F^m &\to F^m\otimes F^n\\
			 e_i\otimes e_j &\mapsto e_j\otimes e_i
		\end{aligned}$
		\item the compact structure $\begin{aligned}[t]
			\eta_{n}: F &\to F^n\otimes F^n\\
			x &\mapsto x\sum_{i=0}^{n-1}e_i\otimes e_i
			\end{aligned}\quad \text{and}\quad
			\begin{aligned}[t]
			\epsilon_{n}: F^n\otimes F^n &\to F\\
			e_i\otimes e_j &\mapsto \braket{e_i}{e_j}
			\end{aligned}$
	\end{itemize}
	
	This category comes with a few additional axioms, that, for simplicity, we will express in graphical form in the following.
	
	This category is represented with string diagrams using the following notations:
	\def\fig{generators}
	\begin{itemize}
		\item Identity $I_{n}:F^n\to F^n$:\qquad $\input{./figures/\fig/\fig_14.tikz}$
		\item Arbitrary tensor $T : F^{n_1}\otimes ... \otimes F^{n_k} \to F^{m_1}\otimes ... \otimes F^{m_\ell}$:\qquad $\input{./figures/\fig/\fig_03.tikz}$
		\item Symmetry $\sigma_{n,m}:F^n\otimes F^m \to F^m\otimes F^n$:\qquad $\input{./figures/\fig/\fig_07.tikz}$
		\item Compact structure: \quad $\eta_n:F\to F^n\otimes F^n:\quad\input{./figures/\fig/\fig_05.tikz}$\quad and \quad $\epsilon_{n}:F^n\otimes F^n\to F:\quad\input{./figures/\fig/\fig_06.tikz}$
		\item Isomorphism $\iota_{n_1,...,n_k}:F^{n_1}\otimes ... \otimes F^{n_k}\to F^{n_1...n_k} $:\qquad $\input{./figures/\fig/\fig_04.tikz}$
		\item Compositions:\qquad $\input{./figures/\fig/\fig_10.tikz}\eq{}\input{./figures/\fig/\fig_11.tikz}$\qquad $\input{./figures/\fig/\fig_12.tikz}\eq{}\input{./figures/\fig/\fig_13.tikz}$
	\end{itemize}
	In the following (and already in the compositions above), we way drop the wire annotations to avoid clutter. In the special case of $I_1$ the identity on $F$, we may represent it with a dashed line.
	As particular cases of tensors, a matrix $M$, a column vector $v$ and a row vector $v^\intercal$ are represented respectively as follows: $\begin{tikzpicture}
	\begin{pgfonlayer}{nodelayer}
		\node [style=none] (230)  at (-2.5, 1.562) {};
		\node [style=none] (231)  at (-2.5, 1.062) {};
		\node [style=none] (232)  at (-3.0, 1.062) {};
		\node [style=none] (233)  at (-3.0, 1.562) {};
		\node [style=none] (234)  at (-2.75, 1.312) {$C_1$};
		\node [style=none] (237)  at (-2.0, 0.688) {};
		\node [style=none] (238)  at (-3.25, 0.688) {};
		\node [style=none] (239)  at (-1.5, 1.312) {};
		\node [style=none] (240)  at (-1.5, 0.562) {};
		\node [style=none] (241)  at (-2.0, 0.562) {};
		\node [style=none] (242)  at (-2.0, 1.312) {};
		\node [style=none] (243)  at (-1.75, 0.938) {$C_2$};
		\node [style=none] (245)  at (-3.25, 1.312) {};
		\node [style=none] (249)  at (-2.5, 1.312) {};
		\node [style=none] (250)  at (-2.0, 1.188) {};
		\node [style=none] (251)  at (-1.5, 0.938) {};
		\node [style=none] (252)  at (-1.25, 0.938) {};
		\node [style=none, rotate=-30] (253)  at (-1.0, 0.688) {...};
		\node [style=none] (254)  at (-0.5, -0.062) {};
		\node [style=none] (255)  at (-3.25, -0.062) {};
		\node [style=none] (256)  at (-0.0, 0.562) {};
		\node [style=none] (257)  at (-0.0, -0.188) {};
		\node [style=none] (258)  at (-0.5, -0.188) {};
		\node [style=none] (259)  at (-0.5, 0.562) {};
		\node [style=none] (260)  at (-0.25, 0.188) {$C_{k}$};
		\node [style=none] (261)  at (-0.75, 0.438) {};
		\node [style=none] (262)  at (-0.5, 0.438) {};
		\node [style=none] (263)  at (0.5, -0.438) {};
		\node [style=none] (264)  at (-4.0, -0.438) {};
		\node [style=none] (265)  at (1.25, 0.188) {};
		\node [style=none] (266)  at (1.25, -0.562) {};
		\node [style=none] (267)  at (0.5, -0.562) {};
		\node [style=none] (268)  at (0.5, 0.188) {};
		\node [style=none] (269)  at (0.875, -0.188) {$C_{k+1}$};
		\node [style=none] (271)  at (-0.0, 0.188) {};
		\node [style=none] (272)  at (0.5, 0.062) {};
		\node [style=none] (274)  at (-3.0, 1.312) {};
		\node [style=none, rotate=90] (275)  at (-2.5, 0.188) {...};
		\node [style=none] (281)  at (1.5, -0.688) {};
		\node [style=none] (283)  at (2.25, -0.062) {};
		\node [style=none] (284)  at (2.25, -0.812) {};
		\node [style=none] (285)  at (1.5, -0.812) {};
		\node [style=none] (286)  at (1.5, -0.062) {};
		\node [style=none] (287)  at (1.875, -0.438) {$C_{k+2}$};
		\node [style=none] (289)  at (1.25, -0.188) {};
		\node [style=none] (290)  at (1.5, -0.188) {};
		\node [style=none] (291)  at (2.25, -0.438) {};
		\node [style=none] (292)  at (2.5, -0.438) {};
		\node [style=none, rotate=-30] (293)  at (2.75, -0.562) {...};
		\node [style=none] (303)  at (3.25, -1.188) {};
		\node [style=none] (304)  at (1.5, -1.188) {};
		\node [style=none] (305)  at (3.75, -0.562) {};
		\node [style=none] (306)  at (3.75, -1.312) {};
		\node [style=none] (307)  at (3.25, -1.312) {};
		\node [style=none] (308)  at (3.25, -0.562) {};
		\node [style=none] (309)  at (3.5, -0.938) {$C_d$};
		\node [style=none] (310)  at (3.0, -0.688) {};
		\node [style=none] (311)  at (3.25, -0.688) {};
		\node [style=none, rotate=90] (313)  at (1.4, -0.938) {...};
		\node [style=none] (316)  at (-3.75, 1.562) {};
		\node [style=none] (317)  at (-3.75, 1.062) {};
		\node [style=none] (318)  at (-3.25, 1.562) {};
		\node [style=none] (319)  at (-3.25, 1.062) {};
		\node [style=none] (320)  at (-3.5, 1.312) {$e_{\ell_1}$};
		\node [style=none] (321)  at (-3.75, 0.188) {};
		\node [style=none] (322)  at (-3.75, -0.312) {};
		\node [style=none] (323)  at (-3.25, 0.188) {};
		\node [style=none] (324)  at (-3.25, -0.312) {};
		\node [style=none] (325)  at (-3.5, -0.062) {$e_{\ell_k}$};
		\node [style=none] (327)  at (-3.75, 0.938) {};
		\node [style=none] (328)  at (-3.75, 0.438) {};
		\node [style=none] (329)  at (-3.25, 0.938) {};
		\node [style=none] (330)  at (-3.25, 0.438) {};
		\node [style=none] (331)  at (-3.5, 0.688) {$e_{\ell_2}$};
		\node [style=tensor, shape border rotate=180] (332)  at (1.125, -0.938) {};
		\node [style=none] (333)  at (1.0, -0.938) {};
		\node [style=none] (334)  at (1.0, -1.562) {};
		\node [style=none] (335)  at (4.0, -1.562) {};
	\end{pgfonlayer}
	\begin{pgfonlayer}{edgelayer}
		\draw (230.center) to (231.center);
		\draw (231.center) to (232.center);
		\draw (232.center) to (233.center);
		\draw (233.center) to (230.center);
		\draw [in=180, out=0] (238.center) to (237.center);
		\draw (239.center) to (240.center);
		\draw (240.center) to (241.center);
		\draw (241.center) to (242.center);
		\draw (242.center) to (239.center);
		\draw (245.center) to (274.center);
		\draw [in=360, out=180] (250.center) to (249.center);
		\draw (252.center) to (251.center);
		\draw [in=180, out=0] (255.center) to (254.center);
		\draw (256.center) to (257.center);
		\draw (257.center) to (258.center);
		\draw (258.center) to (259.center);
		\draw (259.center) to (256.center);
		\draw (262.center) to (261.center);
		\draw [in=360, out=180] (263.center) to (264.center);
		\draw (265.center) to (266.center);
		\draw (266.center) to (267.center);
		\draw (267.center) to (268.center);
		\draw (268.center) to (265.center);
		\draw [in=360, out=180] (272.center) to (271.center);
		\draw (283.center) to (284.center);
		\draw (284.center) to (285.center);
		\draw (285.center) to (286.center);
		\draw (286.center) to (283.center);
		\draw (290.center) to (289.center);
		\draw (292.center) to (291.center);
		\draw [in=0, out=-180] (303.center) to (304.center);
		\draw (305.center) to (306.center);
		\draw (306.center) to (307.center);
		\draw (307.center) to (308.center);
		\draw (308.center) to (305.center);
		\draw (311.center) to (310.center);
		\draw (316.center) to (317.center);
		\draw (317.center) to (319.center);
		\draw (318.center) to (316.center);
		\draw (319.center) to (318.center);
		\draw (321.center) to (322.center);
		\draw (322.center) to (324.center);
		\draw (323.center) to (321.center);
		\draw (324.center) to (323.center);
		\draw (327.center) to (328.center);
		\draw (328.center) to (330.center);
		\draw (329.center) to (327.center);
		\draw (330.center) to (329.center);
		\draw (332) to (333.center);
		\draw [in=180, out=-45] (332) to (304.center);
		\draw [in=-180, out=45] (332) to (281.center);
		\draw [bend right=90, looseness=1.25] (333.center) to (334.center);
		\draw (335.center) to (334.center);
	\end{pgfonlayer}
\end{tikzpicture}$, $\begin{tikzpicture}
	\begin{pgfonlayer}{nodelayer}
		\node [style=none] (203)  at (0.413, 0.188) {};
		\node [style=none] (204)  at (0.413, -0.312) {};
		\node [style=none] (205)  at (1.212, -0.312) {};
		\node [style=none] (206)  at (1.212, 0.188) {};
		\node [style=none] (207)  at (0.812, -0.062) {$U_{k+1}$};
		\node [style=none] (208)  at (1.212, -0.062) {};
		\node [style=none] (209)  at (1.462, -0.062) {};
		\node [style=none] (210)  at (0.163, 0.312) {};
		\node [style=none] (211)  at (0.163, -0.438) {};
		\node [style=none] (212)  at (-0.587, -0.438) {};
		\node [style=none] (213)  at (-0.587, 0.312) {};
		\node [style=none] (214)  at (-0.212, -0.062) {$C_{k+1}$};
		\node [style=none] (215)  at (0.163, -0.062) {};
		\node [style=none] (216)  at (0.413, -0.062) {};
		\node [style=none] (217)  at (-0.587, -0.312) {};
		\node [style=none] (219)  at (-0.587, 0.188) {};
		\node [style=none] (222)  at (-1.462, -0.312) {};
		\node [style=none] (224)  at (-0.737, 0.438) {};
		\node [style=none] (225)  at (-0.737, -0.062) {};
		\node [style=none] (226)  at (-1.287, -0.062) {};
		\node [style=none] (227)  at (-1.287, 0.438) {};
		\node [style=none] (228)  at (-1.012, 0.188) {$v_{k}$};
		\node [style=none] (229)  at (-0.737, 0.188) {};
	\end{pgfonlayer}
	\begin{pgfonlayer}{edgelayer}
		\draw (203.center) to (204.center);
		\draw (204.center) to (205.center);
		\draw (205.center) to (206.center);
		\draw (206.center) to (203.center);
		\draw (209.center) to (208.center);
		\draw (210.center) to (211.center);
		\draw (211.center) to (212.center);
		\draw (212.center) to (213.center);
		\draw (213.center) to (210.center);
		\draw (216.center) to (215.center);
		\draw (222.center) to (217.center);
		\draw (224.center) to (225.center);
		\draw (225.center) to (226.center);
		\draw (226.center) to (227.center);
		\draw (227.center) to (224.center);
		\draw [in=180, out=0] (229.center) to (219.center);
	\end{pgfonlayer}
\end{tikzpicture}$ and $\input{./figures/\fig/\fig_02.tikz}$.
	The compact structure allows us to define the transpose of tensors, as:
	$\input{./figures/\fig/\fig_15.tikz}\eq{}\input{./figures/\fig/\fig_16.tikz}$\\
	Interestingly, the inverse of the isomorphism $\iota$ is its transpose. We hence may define it as follows:	$\input{./figures/\fig/\fig_08.tikz}~~:=~~\input{./figures/\fig/\fig_09.tikz}$. This allows us to define the row vectorised form of an arbitrary tensor:
	\tikzfig{row-vectorisation-2}.
	This also allows us to define the Kronecker product of two matrices $M$ and $N$ can as a reshaping of their tensor product: \def\fig{TN-tensor-1}$$	
	
	The graphical notation makes the bifunctorial law a tautology, as each side of the equation is represented in the exact same way as string diagrams. The additional axioms of compact-close symmetric monoidal categories (satisfied by $\cat{Mat}_F$) are the following, with $T$ an arbitrary tensor:
	\begin{align*}
		\tikzfig{TN-identity}&\qquad\qquad
		\tikzfig{TN-swap-naturality}\qquad\qquad
		\tikzfig{TN-symmetry}\\[1em]
		\tikzfig{TN-swap-cap}&\qquad\qquad
		\tikzfig{TN-snake}
	\end{align*}
	A crucial result on string diagrams for compact-close symmetric monoidal categories, is that one can deform any diagram at will, using the above axioms:
	\begin{theorem}[\cite{Selinger2010Survey}]
		The isomorphism of diagrams follows from the axioms of compact-close symmetric monoidal categories. Two isomorphic diagrams are hence equal in such a theory.
	\end{theorem}
	These axioms don't account for the isomorphism $\iota$, which is specific to the category $\cat{Mat}_F$. We hence introduce the following additional equations to axiomatise it:
	\begin{align*}
		\tikzfig{TN-assoc}\qquad\qquad
		\tikzfig{TN-tensor-2}
	\end{align*}
	It is fairly easy to check that these are indeed satisfied. We do not claim any form of completeness here, and simply presented the equations that are necessary for the following.
	
	Within this framework, we may redefine what a tensor train decomposition is. To simplify the presentation, we will work on states. Notice that thanks to the compact closure, tensor train decompositions for states naturally generalise to tensor train decompositions for arbitrary tensors.
	
	\begin{definition}
		A tensor train decomposition of $T:F\to F^{n_1}\otimes ... \otimes F^{n_d}$ is a decomposition of the form:
		\[\tikzfig{tensor-train}\]
	\end{definition}

	Notice that $C_1$ and $C_d$ are represented differently, because $r_0=r_d=1$. We could have represented them as $\tikzfig{C1}$ and $\tikzfig{Cd}$ respectively to keep them consistent with the rest, albeit at the cost of a more cluttered form.
	
	Before moving on, we want to emphasise the importance of this graphical framework for the following: it is completely formal, and it makes the handling of tensors, and all proofs that manipulate them, much easier and cleaner to present. Without them, proofs can quickly look like a mess with large amounts of indices to properly keep track of during reshapes.
	
	\subsection{Reduction}
	
	Tensor train decompositions are clearly not unique, since, for any matrix $A$ with a left inverse $A^{-1}$, we have:
	\begin{align*}
		\tikzfig{TT-decomp-not-unique}
	\end{align*}
	by pushing $A$ in $C_2$ and $A^{-1}$ in $C_1$.
	
	In order to reduce as much as possible the memory size required to represent a tensor by a tensor train, we want to find a decomposition that minimises the TT-ranks. One way to do this is to exploit the property above, using rank factorisations of the cores. Such a decomposition has been identified as the \emph{Hierarchical Singular Value Decomposition (HSVD)}. Importantly, such decomposition can be carried out on an existing tensor train, with a complexity that is polynomial in $d$, $r$ and $n$.
	
	While the SVD is numerically stable, provides the best low-rank approximations of arbitrary matrices, and is amenable to several optimisation techniques from numerical methods \cite{Baboulin2025Mixed}, the HSVD suffers from a few drawbacks:
	\begin{itemize}
		\item It is not unique, although several variations have been introduced to constrain the TT form \cite{Holtz2011Manifolds,Perez2006mps,Vidal2003Efficient}
		\item Singular Value Decompositions are only defined on $\mathbb R$ and $\mathbb C$, and can hence only be approximated in practice
		\item There is no known efficient way of computing the first (or really any) index that contains a value with amplitude $>\epsilon$ for a given threshold $\epsilon\geq0$
	\end{itemize}
	
	The rest of the paper is devoted to providing a new tensor train decomposition that is valid for any field-valued tensor, and reaches optimal TT-ranks in the exact computation case. This decomposition is based on a simpler matrix factorisation than the SVD: the row-echelon form, allowed by mere Gauss-Jordan elimination. As we will see, the tensor train decomposition we obtain is unique, and can be reached from any tensor train in time polynomial in its size.
	
	\section{LDPU Decomposition}
	\label{sec:LDPU}
	
	In the following, we will use the usual notion of row-echelon form (REF) of a matrix, whereby for every non-zero row, its leading index is (strictly) larger than the leading indices of the non-zero rows above\footnote{REF matrices are usually described as having their zero rows at the bottom. We do not impose this here, for almost everywhere in the following, the zero-rows are going to be removed anyway.}. 
	The row-echelon form process may easily give a rank decomposition of the initial matrix, by removing the zero rows in the echelon form, and the corresponding columns in the transformation matrix that yields the matrix in REF. Since in such a matrix, the rows are linearly independent, we say that such a REF matrix is full row-rank. If the leading coefficients of all the non-zero rows of a matrix in REF are $1$, we say that it is in \emph{unit} REF. Obviously, we say that a matrix is in (unit) column-echelon form (CEF) if its transpose is in (unit) row-echelon form.
	
	\subsection{Algorithm and uniqueness}
	
	A standard way to compute a REF out of a matrix is to compute a PLU decomposition~\cite{Golub2013}. This one is however not unique, hence we will use a much less common variation, called LDPU decomposition~(which can be partially found in~\cite{Carrell2017} for square matrices). To do so, we need an extra notion on column-echelon form matrices:
	\begin{definition}
		Let $L\in F^{n\times r}$ be a rank-$r$ matrix in column-echelon form, and $P\in F^{r\times r}$ be the matrix representing the permutation $\sigma$. Let $(r_k)_{0\leq k<r}$ be the sorted sequence of $L$'s pivots' row indices. We say that $L$ is $P$-conditioned if:
		\[\forall i,j\in\{0,...,r-1\},~\sigma^{-1}(j) > \sigma^{-1}(i) \implies L[r_i,j]=0\]
	\end{definition}
	Notice that since $L$ is in column-echelon form, we already have $i<j\implies L[r_i,j]=0$. Being $P$-conditioned hence only adds further constraints when $i>j$ and $\sigma^{-1}(j) > \sigma^{-1}(i)$, i.e.~when $\sigma^{-1}$ inverses $i$ and $j$.
	\begin{definition}
		Let $A\in F^{n\times m}$ be a matrix of rank $r$. An \emph{LDPU decomposition} of $A$ is of the form $A = LDPU$ where:
		\begin{itemize}
			\item $U\in F^{r\times m}$ is in unit row-echelon form
			\item $P\in F^{r\times r}$ is a permutation matrix
			\item $D\in F^{r\times r}$ is diagonal
			\item $L\in F^{n\times r}$ is in unit column-echelon form and $P$-conditioned
		\end{itemize}
	\end{definition}
	Notice that since $r$ is the rank of $A$, $U$ has to be full row-rank (i.e.~the rows of $U$ form a linearly independent family of row vectors), and similarly $L$ has to be full column-rank. Notice also that when working in $\mathbb F_2$, $D$ is merely the identity matrix, and can hence be ignored. 
	
	\begin{example}
		\renewcommand{\arraystretch}{0.85}
		\setlength{\arraycolsep}{2pt}
		Here is an LDPU of the matrix on the left, over $\mathbb Q$:
		\begin{align*}
			\begin{pmatrix}
				0 & 0 & 1 & 2 \\
				1 & 2 & 3 & 4 \\
				1 & -1 & 1 & -1 \\
				2 & 1 & 4 & 3
			\end{pmatrix}
			= \begin{pmatrix}
				1 & 0 & 0 \\
				0 & 1 & 0 \\
				0 & 1 & 1 \\
				0 & 2 & 1
			\end{pmatrix}
			\begin{pmatrix}
				1 & 0 & 0 \\
				0 & 1 & 0 \\
				0 & 0 & -3
			\end{pmatrix}
			\begin{pmatrix}
				0 & 0 & 1 \\
				1 & 0 & 0 \\
				0 & 1 & 0
			\end{pmatrix}
			\begin{pmatrix}
				1 & 2 & 3 & 4 \\
				0 & 1 & \frac{2}{3} & \frac{5}{3} \\
				0 & 0 & 1 & 2
			\end{pmatrix}
			=LDPU
		\end{align*}
		With $\sigma$ the permutation represented by $P$, we see that $\sigma^{-1}(0)=2$, $\sigma^{-1}(1)=0$ and $\sigma^{-1}(2)=1$. We have $\sigma^{-1}(0)>\sigma^{-1}(1)$ hence $L[r_1,0] = L[1,0] = 0$. Similarly, $L[r_2,0]=L[2,0] =0$. As a result, $L$ is indeed $P$-conditioned.
	\end{example}
	
	Obviously, an L\textbf{DP}U decomposition directly provides an L\textbf{PD}U decomposition, by pushing the diagonal matrix through the permutation matrix.
	
	It is fairly easy to get an LDPU decomposition of any matrix, by adapting the standard Gaussian elimination algorithm. First we make sure that all the rows of $A$ are either null or have unique pivots, by performing as few operations on the matrix as possible: when considering a row, we add another previous row to it only if its leading index is a previously identified pivot. After setting all leading coefficients to $1$, we get a matrix which is in REF up to a permutation of the rows, which will constitute the permutation matrix $P$. Since we only add given rows to rows below them, the obtained transformation matrix $L$ is lower triangular. Moreover, a row is never added to a row with a smaller leading index. This is captured by $L$ being $P$-conditioned. We can then make it unit by extracting the diagonal coefficients, that we store in the diagonal matrix $D$. Finally, for every zero row in the REF matrix, we remove it, as well as the corresponding row/column in the other three matrices. The obtained algorithm is more formally expressed in \Cref{alg:LDPU}, where $\mathbf{0}_{n\times m}$ represents the zero $n\times m$ matrix.

	Notice that if the rank of the matrix is $0$ (i.e.~if the matrix is null), the algorithm still produces a result, with $L$ having no column, $U$ having no row, and $D$ and $P$ being $0\times0$ matrices. 
	The existence of this algorithm ensures the existence of an LDPU decomposition for any matrix. This decomposition also turns out to be unique:
	\begin{proposition}
		\label{prop:LDPU-uniqueness}
		For any matrix $A\in F^{n\times m}$, the LDPU decomposition of $A$ exists and is unique. It can be computed in time $\mathcal O(nm\min(n,m))$.
	\end{proposition}
	
	As usual, everything is expressed here as a matrix, to simplify the presentation, although in practice, some elements may be more compactly stored (in particular the diagonal and permutation matrices). We are not concerned here with the efficiency and implementation details of the algorithm, for we are more interested in the theoretical results that the decomposition allows.
	\begin{remark}
		This LDPU decomposition, and subsequently, the tensor train decomposition developed in this article, likely suffer from: i/ numerical instability when working with floating point arithmetic, and ii/ exponential bit complexity when working with exact representations of values in infinite fields. However, those two drawbacks disappear when working with finite fields such as $\mathbb F_2$.
	\end{remark}
	
	\subsection{Properties of row-echelon forms}
	
	In the following, we will need several results about matrices in row-echelon form, that we state here. The first state that the REF is preserved by Kronecker product:
	
	\begin{lemma}
		\label{lem:REF-tensor-id}
		Let $U_0$ and $U_1$ be two full row-rank matrices in (unit) row-echelon form. Then their Kronecker product $\tikzfig{REF-tensor}$ is in full row-rank (unit) row-echelon form.
	\end{lemma}
	Notice that the identity matrices are full row-rank unit REF matrices, so the lemma applies as a particular case when either $U_0$ or $U_1$ is replaced by the identity. The next lemma states that the REF is also preserved by composition:
	\begin{lemma}
		\label{lem:REF-compo}
		The product of two matrices in row-echelon form is in row-echelon form. If the two matrices are unit, the result is also unit.
	\end{lemma}
	The next lemma addresses the existence of right inverses for matrices in REF. It is mostly used to prove the two lemmas afterwards.
	\begin{lemma}
		\label{lem:REF-right-inverse}
		A (unit) row-echelon form matrix has a right inverse, iff it is full row-rank. In that case, the right inverses are in (unit) row-echelon form.
	\end{lemma}
	The following lemma refines \Cref{lem:REF-compo} on products of matrices in REF, for the full-rank case:
	\begin{lemma}
		\label{lem:REF-compo-frr}
		The product of two matrices in full row-rank row-echelon form is in full row-rank row-echelon form.
	\end{lemma}
	The final lemma will be necessary for the proof of uniqueness of the normal form, in particular to propagate the uniqueness of individual cores:
	\begin{lemma}
		\label{lem:unique-REF-factorisation}
		Let $U_1$ and $U_2$ be two full row-rank (unit) row-echelon form matrices. If the equation $XU_1=U_2$ has a solution, then it has a unique solution, which is in full row-rank (unit) row-echelon form.
	\end{lemma}
	
	\section{The Normal Form}
	\label{sec:NF}
	
	We start by giving an algorithm that builds a tensor train from an arbitrary tensor for which we have the full description, using the LDPU decomposition described above. Since the algorithm is deterministic, we will define its result as the (unique) normal form of the tensor train decomposition. Defining the normal form from this algorithm might seem unorthodox, but we feel like it is simpler to understand the normal form by trying to build a TT directly from the full tensor. We will later show how to reach that normal form from any given tensor train, without the need to compute the full tensor.
	
	\subsection{Breaking down arbitrary tensors}
	
	To do so, we need three helper functions that essentially reshape the cores/matrices:
	\begin{itemize}
		\item $\operatorname{split}(A,n)$ is a ``partial vectorisation'', it assumes matrix $A$ has $kn$ columns for some $k$, and breaks each row of $A$ into $n$ pieces of size $k$, to build a block of those $n$ smaller rows in order, matricially and graphically:
		\[
		\renewcommand{\arraystretch}{0.8}
		\setlength{\arraycolsep}{2pt}
		\operatorname{split}\left[
		\begin{pmatrix}
			a_{1,1}\cdot \ldots \cdot a_{1,n}\\
			\vdots\\
			a_{\ell,1}\cdot \ldots \cdot a_{\ell,n}
		\end{pmatrix}
		,n
		\right]
		= 
		\scalebox{0.8}{$
		\begin{pmatrix}
			a_{1,1}\\
			\vdots\\
			a_{1,n}\\
			\scalebox{1.3}{\textbf{$\vdots$}}\\
			a_{\ell,1}\\
			\vdots\\
			a_{\ell,n}\\			
		\end{pmatrix}$}
		\hspace*{3em}
		\tikzfig{split}
		\]
		where $a_{i,j}$ are row-vectors of size $k$, and $(\_ \cdot\_)$ is the concatenation of row vectors. 
		\item $\operatorname{untwine}(L,n,P)$ assumes $L$ has $nk$ rows and $\ell$ columns for some $k$ and $\ell$, and that $P$ is a permutation matrix of size $k$. The result is a core $C\in F^{n\times k\times \ell}$ such that 
		$PC[i] = L[j=i \bmod n]$, where $L[j=i \bmod n ]$ is matrix $L$ where we only keep the rows whose indices verify the condition. Graphically:
		\[\tikzfig{untwine-2}\]
		\item $\operatorname{intertwine}$ is the reverse operation:  $\operatorname{intertwine}(\operatorname{untwine}(L,n,P),P)=L$, graphically:
		\[\tikzfig{intertwine-2}\]
	\end{itemize}
	
	The algorithm is given in \Cref{alg:NF-from-full-tensor}. For a full tensor $T\in F^{n_1\times ... \times n_d}$, it is to be called on the row-vectorised version of $T$ (described as $\vec T^\intercal$ above), the list $[n_1,...,n_d]$, an empty list for the accumulator, and the trivial permutation matrix $\begin{pmatrix}1\end{pmatrix}$.
	
	\begin{algorithm2e}[!ht]
		\caption{Normal form from full tensor}
		\label{alg:NF-from-full-tensor}
		\KwData{A list of integers $\ell = [n_1,...,n_d]$, a matrix $A$ with $\prod_i n_i$ columns, an accumulator $C$ (containing the resulting list of cores), a permutation $P$.}
		\KwResult{Recursively builds the list of cores that constitutes the tensor train for the starting tensor.}
		$n \gets n_1$\;
		$\ell\gets [n_2,...,n_d]$\;
		$A'\gets \operatorname{split}(A,n)$\;
		$L,D,P',U\gets \operatorname{LDPU}(A')$\;
		$M\gets LDP'$\;
		$c\gets \operatorname{untwine}(M,n,P)$\;
		append $c$ to $C$\;
		\eIf{$l$ is empty}{\Return{$C$}}{
			Recursively call the algorithm on $P'U$, $\ell$, $C$, $P'$
		}
	\end{algorithm2e}

	We can graphically interpret this algorithm using the tensor network notation, in \Cref{fig:algo-interp}, using $P_d=U_d=\begin{pmatrix}1\end{pmatrix}$. This graphical understanding will prove useful for the upcoming proofs.
	
	\begin{figure*}[!ht]
		\def\fig{algo-interpretation-left-right}
		\begin{align*}
			\scalebox{0.75}{$
			\input{./figures/\fig/\fig_00.tikz}$}
			&\scalebox{0.75}{$\eq{}\input{./figures/\fig/\fig_01.tikz}
			\eq{}\input{./figures/\fig/\fig_02.tikz}
			\eq{}\input{./figures/\fig/\fig_03.tikz}
			\eq{}\input{./figures/\fig/\fig_04.tikz}$}\\
			&\scalebox{0.75}{$\eq{}\input{./figures/\fig/\fig_05.tikz}
			\eq{}...
			\eq{}\input{./figures/\fig/\fig_06.tikz}$}\\
			&\scalebox{0.75}{$\eq{}\input{./figures/\fig/\fig_07.tikz}$}
		\end{align*}
		\caption{Graphical interpretation of \Cref{alg:NF-from-full-tensor}}
		\label{fig:algo-interp}
	\end{figure*}
	
	An example of the application of the algorithm can be found in \Cref{ex:NF}, with steps detailed in the appendix.
	
	\begin{figure}[!ht]
		\centering
		$\vec T^\intercal = \left(\begin{array}{rrrrrrrrrrrr}
			0 & 0 & 1 & -1 & 1 & 2 & 2 & -1 & -1 & 1 & 0 & 2
		\end{array}\right)$\\
		$\begin{array}{ll}
			c_1[0]=\begin{pmatrix}0&1\end{pmatrix}&
			c_2[0]=\begin{pmatrix}1&\frac12&0\\0&0&1\end{pmatrix} 
			\\
			c_1[1]=\begin{pmatrix}2&0\end{pmatrix}&
			c_2[1]=\begin{pmatrix}\frac12&\frac12&\frac12\\-1&0&0\end{pmatrix}
		\end{array}
		c_3[0]=\begin{pmatrix}1\\0\\0\end{pmatrix}~~ c_3[1]=\begin{pmatrix}-1\\1\\0\end{pmatrix}~~
		c_3[2]=\begin{pmatrix}-2\\3\\1\end{pmatrix}$
		\caption{The row-vector form of a tensor $T \in \mathbb Q^{2\times2\times3}$, and the result of \Cref{alg:NF-from-full-tensor} applied to it. The steps are detailed in the appendix.}
		\label{ex:NF}
	\end{figure}

	\subsection{Unique characterisation}
	
	The tensor train obtained from this algorithm has two main properties, captured by the following notion of normal form:
	
	\begin{definition}
		Let $C=(C_1,...,C_d)$ be a tensor train with modes $(n_1,...,n_d)$ and TT-ranks $(r_1,...,r_d)$. We say that $C$ is in normal form if there exist permutation matrices $P_0= \begin{pmatrix}1\end{pmatrix}, P_1,...,P_d$, such that for all $1\leq i\leq d$:
		\begin{itemize}
			\item the block matrix $\begin{pmatrix}C_i[0]\mid \ldots \mid C_i[n_i-1]\end{pmatrix}$ is in full row-rank row-echelon form, unit if $i>1$ 
			\item $\operatorname{intertwine}(C_i, n_i, P_{i-1})P_i^{-1}$ is in full column-rank column-echelon form and $P_i$-con\-ditioned
		\end{itemize}
	\end{definition}
	Notice that these conditions can be checked in polynomial time in the size of the tensor train.
	
	Notice also that for $i=1$, the first condition tells us that either $r_1=0$, or $C_1$ is non-null, for $C_1[k]$ is a row vector for all $k$, hence $\begin{pmatrix}C_1[0]\mid \ldots\end{pmatrix}$ is obviously in row-echelon form unless it is null. Similarly, for $i=d$, the second condition is obviously met unless $C_d$ is null.
	
	Graphically, the two conditions can be expressed as:
	\begin{itemize}
		\item $\tikzfig{NF-REF-2}$ is in row-echelon form (unit for $i>1$) and full row-rank
		\item $\tikzfig{NF-CEF-2}$ is in column-echelon form, full column-rank, and $P_i$-conditioned
	\end{itemize}
	We call this form \emph{normal}, for, as we will see, it can be obtained by the upcoming reduction of tensor trains.
	
	\begin{proposition}
		\label{prop:algo-in-NF}
		The tensor train obtained by \Cref{alg:NF-from-full-tensor} is in normal form.
	\end{proposition}
	
	\begin{proof}
		One can see that the second condition of normal forms is obviously satisfied, by construction. To show that the first condition is also satisfied, we show by induction, and using \Cref{lem:REF-tensor-id}, that the composition of the last cores, up to a reshape, is in unit row-echelon form. We can then use \Cref{lem:unique-REF-factorisation} to conclude. The full proof can be found in the appendix.
	\end{proof}
	
	It is then possible to show the main result of the paper:
	
	\begin{theorem}
		\label{thm:NF-uniqueness}
		The normal form of tensor trains exists and is unique.
	\end{theorem}
	
	\begin{proof}
		We start by showing by induction that the composition of the last cores up to a reshape is in full row-rank unit row-echelon form, using the first condition on normals forms, as well as Lemmas \ref{lem:REF-tensor-id}, \ref{lem:REF-compo} and \ref{lem:REF-compo-frr}. Similarly, we show by induction that the composition of the first cores, up to a reshape and permutation matrix, is in full column-rank column-echelon form. We can then fix each core one after the other by showing that the cut between the first cores and the rest corresponds to the (unique) LDPU decomposition of the full tensor (up to a reshape). The full proof can be found in the appendix.
	\end{proof}
	
	\section{Reduction of Tensor Trains}
	\label{sec:reduc}
	
	We show in this section how to turn any tensor train in normal form, by only performing operations on the tensor train itself, not on the full tensor, resulting in an algorithm that is polynomial in the size of the tensor train, and not of the full tensor. This can be seen as the equivalent for tensor trains of the usual (weighted) BDD reductions.
	
	Just like the HSVD form, the algorithm will happen in two sweeps, the first from right to left, and the second from left to right, but the two will be slightly different (although at their core they perform a rank-factorisation to reduce the size of the tensor train). We will show that the sweeps cannot augment the size of the tensor train, and hence that normal form minimises the size of the tensor train (for the given ordering of the modes). The first and second sweeps essentially enforce that the tensor train verifies the first and second normal form conditions, respectively.
	
	\subsection{First sweep}
	
	The first sweep is the simplest: for each core starting from the last, we make sure that $\tikzfig{NF-REF-2}$ is in unit full row-rank row-echelon form by computing its row-echelon form\footnote{This can be done by the LDPU decomposition, or, in this case, more simply by the PLU decomposition. We stick here with the LDPU decomposition to remain consistent with the rest of the paper.}:
	\def\fig{sweep-1-details}
	\begin{align*}
		
		\eq{}
		\eq{}\input{./figures/\fig/\fig_02.tikz}
		\eq{}\input{./figures/\fig/\fig_03.tikz}
	\end{align*}
	The new core $C_i'$ is then a reshape of $U$, while $C_{i-1}$ is updated by pushing the matrix $LDP$ into its input. 
	It should be clear that this operation does not modify the full tensor. The algorithm for the first sweep is given in more detail in \Cref{alg:sweep-1}.
	
	\begin{algorithm2e}[!ht]
		\caption{First sweep for TT reduction}
		\label{alg:sweep-1}
		\KwData{A tensor train $(C_1,...,C_d)$ with modes $n_1,...,n_d$.}
		\KwResult{Reduces the tensor train, to make it compliant with Condition 1 of the normal form.}
		\For{$k$ from $d$ down to $2$}{
			$L,D,P,U=\begin{pmatrix}
				U_1\mid \ldots\mid U_{n_k}
			\end{pmatrix}\gets \operatorname{LDPU}\left[\begin{pmatrix}
				C_k[0]\mid \ldots\mid C_k[n_k-1]
			\end{pmatrix}\right]$\;
			\For{$i\in\{0,...,n_{k-1}-1\}$}{
				$C_{k-1}[i] \gets C_{k-1}[i]LDP$
			}
			\For{$i\in\{0,...,n_k-1\}$}{
				$C_k[i] \gets U_i$
			}
		}
	\end{algorithm2e}
	
	While this first sweep in itself is not enough to reach the normal form (nor is it enough to get to the minimal size of the tensor train), it is enough to answer the nullity of a tensor train:
	\begin{proposition}
		\label{prop:nullity-checking}
		Let $C$ be a tensor train obtained by the first sweep, with $C_1\in F^{n_1\times 1\times r_1}$ its first core. The full tensor is null iff $C_1$ is null.
	\end{proposition}
	This result can be useful for several purposes, without the need to get to the unique normal form. Indeed, one can check if two tensor trains are equal by subtracting one from the other (done in polynomial time), apply a first sweep to the result and check if it is null. When working in $\mathbb F_2$, one can check the satisfiability of a formula by building a tensor train representation of it, and checking for its nullity in the same manner.
	
	Related to satisfiability, one may want to enumerate, or at the very least extract a solution, to the satisfiability of a formula when it is indeed satisfiable. More generally, it is possible to get the index and value of the first non-zero element of a tensor train obtained from the first sweep.
	\begin{proposition}
		\label{prop:leading-index-coeff}
		Let $C$ be a non-null tensor train of order $d$, maximum mode dimension $n$ and maximum TT-rank $r$, obtained by the first sweep. The leading index and leading coefficient of the full tensor can be computed in time $\mathcal O(dnr)$. An algorithm is given in \Cref{alg:leading-index-coeff}.
	\end{proposition}
	Notice that although the algorithm computes a coefficient of the tensor (the leading one), it does so in $\mathcal O(dnr)$, which is often better than the $\mathcal O(dr^2)$ required to get a coefficient at an arbitrary index -- by doing $d$ matrix-vector products. This is thanks to the echelon structure of the cores, which allows us, for the leading coefficient, to avoid unnecessary computations.
	
	\subsection{Second sweep and full reduction}
	
	Although very useful, the first sweep does not yield a unique normal form, nor does it even minimise the TT-ranks.
	
	The second sweep fixes this issue, and is similar in essence, but may seem less intuitive because of the permutations that we have to keep track of, although the graphical notation should make it clearer. We start this time from the first core with trivial permutation $P_0$ on one element, and enforce the core's compliance with the second condition of normal forms, by updating the core as follows, assuming permutation $P$:
	\def\fig{sweep-2-details}
	\begin{align*}
		
		\eq{}
		\eq{}\input{./figures/\fig/\fig_02.tikz}
		\eq{}\input{./figures/\fig/\fig_03.tikz}
	\end{align*}
	The core $C_{i+1}$ is then updated by pushing $U$ into its first mode. 
	we then continue with the next core, using the newly computed $P'$ as the permutation matrix. The algorithm for the second sweep is given in detail in \Cref{alg:sweep-2}.
	
	\begin{algorithm2e}[!ht]
		\caption{Second sweep for TT reduction}
		\label{alg:sweep-2}
		\KwData{A tensor train $(C_1,...,C_d)$ with modes $n_1,...,n_d$.}
		\KwResult{Reduces the tensor train, to make it compliant with Condition 2 of the normal form.}
		$P \gets \begin{pmatrix}1\end{pmatrix}$\;
		\For{$k$ from $1$ to $d-1$}{
			$L,D,P',U\gets \operatorname{LDPU}(\operatorname{intertwine}(C_k,P))$\;
			\For{$i\in\{0,...,n_{k+1}-1\}$}{
				$C_{k+1}[i] \gets UC_{k+1}[i]$
			}
			$C_k\gets \operatorname{untwine}(LDP',n_k,P)$\;
			$P\gets P'$
		}
	\end{algorithm2e}
	
	\begin{theorem}
		\label{thm:reduction}
		The algorithm consisting of the first sweep followed by the second sweep sets any tensor train in normal form. The algorithm runs in $\mathcal O(rs)$ with $r$ the maximum TT-rank, and $s$ the size of the tensor train.
		The resulting tensor train has size at most the size of the starting tensor train.
	\end{theorem}
	
	Notice that if at some point a null core is created, the nullity is spread from one core to the next during the sweeps. When applying the algorithm to a tensor train representing a zero tensor, after the 1st sweep, the first core will be null, by \Cref{prop:nullity-checking}. The second sweep then makes every core null, and the TT-ranks (except $r_0$ and $r_d$) end up being $0$.
	
	Since the size of a tensor train is smaller that $dnr^2$ with $d$ number of modes, $n$ its largest mode, and $r$ its largest TT-rank, we can re-express more roughly the runtime of the algorithm as $\mathcal O(dnr^3)$. To compute it, we assumed naive matrix-matrix multiplication, as well as the unoptimised LDPU decomposition described above. If the usual LU decomposition were used instead of the LDPU here, the complexity could be brought to $\mathcal O(dnr^\omega)$ where $\omega$ is the currently best-known power such that matrix-matrix multiplication can be done in $\mathcal O(n^\omega)$. It is currently not known whether the LDPU decomposition can be done in $\mathcal O(n^\omega)$ time, as is the case for the PLU decomposition. This is left as an open question. In the case of $\mathbb F_2$, it might be possible to accelerate the average-case complexity by adapting the so-called ``method of 4 Russians'' \cite{Albrecht2010Decomposition} to get the LDPU decomposition. This is again left as an open question.
	
	As a direct consequence of the last part of the theorem:
	\begin{corollary}
		The normal form of tensor trains minimises the TT-ranks and size.
	\end{corollary}
	
	In the worst cases, the size of a tensor train decomposition will remain exponential with respect to the order of the tensor. This is akin to BDDs, where, in the worst case, the number of edges is exponential in the number of variables $n$. In the latter case, however, the number of edges is no more than twice the number of valuations of its variables (i.e.~$2\cdot2^n$ for $n$ variables). We may wonder whether we can similarly upper-bound the size of a reduced tensor train decomposition, with respect to the total number of entries in the tensor, i.e.~$n_1\cdot...\cdot n_d$. It turns out that here as well, there is only a small constant multiplicative overhead in the worst case:
	\begin{proposition}
		\label{prop:TT-upper-bound}
		Let $T\in F^{n_1\times ... \times n_d}$ be a fully reduced tensor of order $d$. The size $s$ of $T$ is bounded above as: $\displaystyle s\leq \frac83\prod_{i=1}^d n_i$.
		
		This bound can be reached:
		There exists $T\in F^{n_1\times ... \times n_d}$ with $n_1=...=n_d=2$, such that its size is $\displaystyle s = \frac83\prod_{i=1}^d n_i - \mathcal O(1)$.
	\end{proposition}
	
	\section{Conclusion}
	
	We have presented here a normal form and a reduction strategy that minimises tensor trains over arbitrary fields in polynomial time. This strategy is based on the LDPU decomposition defined above, and for which we gave a naive implementation. This implementation is akin to the usual Gaussian elimination, and hence likely suffers from two drawbacks: i/ numerical instability when working with floating point arithmetic, and ii/ exponential bit complexity when working with exact representations of values in infinite fields. Those two drawbacks disappear when working with finite fields, such as $\mathbb F_2$, when one wants to generalise decision diagrams. Addressing these two issues in infinite fields is left as an open problem.
	
	We have also given an algorithm that directly decomposes arbitrary tensors into the unique equivalent tensor train in normal form. This algorithm and the one for reduction, the main proofs of the claims, and the definition of the normal form itself, were all made much simpler thanks to the formal graphical representation of the category $\cat{Mat}_F$ and manipulation of its diagrams through the equational theory. Finally, we've shown that we can efficiently extract the first non-zero value and corresponding index, and that in the worst case, a tensor train uses at most $8/3$ times the space required for a naive storage, akin to the BDD case.

	\section*{Acknowledgments}

	The author acknowledges support from the PEPR integrated project EPiQ ANR-22-PETQ-0007 part of Plan France 2030, the ANR projects TaQC ANR-22-CE47-0012 and HQI ANR-22-PNCQ-0002.

	
	\appendix
	
	\newpage
	
	\section{Algorithms}
	
	\begin{algorithm2e}[!ht]
		\caption{LDPU decomposition}
		\label{alg:LDPU}
		\KwData{A matrix $A\in F^{n\times m}$}
		\KwResult{Lower unitriangular matrix $L$, diagonal matrix $D$, permutation matrix $P$ and upper unitriangular matrix $U$ such that $A=LDPU$.}
		\If{$A$ is null}{
			\Return{$\mathbf{0}_{n\times0}$, $\mathbf{0}_{0\times0}$, $\mathbf{0}_{0\times0}$, $\mathbf{0}_{0\times m}$}
		}
		$L\gets I_n$\Comment*[r]{$n\times n$ identity matrix}
		$N\gets [~]$\Comment*[r]{empty list, stores indices of null rows of $A$}
		$p\gets \{\}$\Comment*[r]{empty hash table,maps pivot column to corresponding row}
		\For{$i$ from $0$ to $n-1$}{
			$j\gets$ leading index of $A[i]$\Comment*[r]{$m$ if $A[i]=0$}
			\While{$j < m$ and $\exists r,~(j\mapsto r) \in p$}{
				$L[r] \gets L[r] + A[i,j]L[i]$\Comment*[r]{pivot already exists}
				$A[i] \gets A[i] - A[i,j]A[r]$\;
				$j\gets$ leading index of $A[i]$
			}
			\eIf(\hfill \texttt{/* new pivot found */}){$j < m$}{
				add $(j\mapsto i)$ to $p$\;
				$L[i,i]\gets A[i,j]$\;
				$A[i]\gets A[i] / A[i,j]$
			}(\hfill \texttt{/* row is null */}){
				append $i$ to $N$
			}
		}
		$D\gets$ diagonal of $L$\;
		\For{$i\in N$}{
			remove row $i$ of $L$, $D$, $P$ and $A$\;
			remove column $i$ of $D$
		}
		$P\gets$ the permutation matrix such that $P^{-1}A$ is in REF\;
		\Return{$L^\intercal D^{-1}$, $D$, $P$, $P^{-1}A$}
	\end{algorithm2e}
	
	\newpage
	
	\begin{algorithm2e}[!ht]
	\caption{Leading index and coefficient.}
	\label{alg:leading-index-coeff}
	\KwData{A tensor train $(C_1,...,C_d)$ with modes $n_1,...,n_d$, TT-ranks $r_0,...,r_d$, obtained from the first sweep.}
	\KwResult{The leading index and coefficient of the full tensor.}
	$i\gets 0$\;
	$c \gets 1$\;
	$I\gets [~]$\;
	\For{$k$ from $1$ to $d$}{
		$\ell\gets$ first index such that $C_k[\ell,i]$ is not null\;
		append $\ell$ to $I$\;
		$c \gets c~\times$ leading coefficient of $C_k[\ell,i]$\;
		$i\gets$ leading index of $C_k[\ell,i]$
	}
	\Return{$I$, $c$}
	\end{algorithm2e}
	
	\newpage
	
	\section{Details of \Cref{ex:NF}}
	
	We apply the algorithm on :
	\setlength{\arraycolsep}{4pt}
	\[\vec T^\intercal = \left(\begin{array}{rrrrrrrrrrrr}
		0 & 0 & 1 & -1 & 1 & 2 & 2 & -1 & -1 & 1 & 0 & 2
	\end{array}\right),
	\quad [2,2,3],\quad [~],\quad\begin{pmatrix}1\end{pmatrix}\]
	The different steps are then:
	\begin{enumerate}
		\item $\begin{aligned}[t]
			&\operatorname{LDPU}\left[\begin{pmatrix}
				0 & 0 & 1 & -1 & 1 & 2 \\
				2 & -1 & -1 & 1 & 0 & 2
			\end{pmatrix}\right]\\
			&= \lefteqn{\underbracket{\phantom{
						\begin{pmatrix}
							1 & 0 \\
							0 & 1
						\end{pmatrix}
						\begin{pmatrix}
							1 & 0 \\
							0 & 2
						\end{pmatrix}
						\begin{pmatrix}
							0 & 1\\
							1 & 0
				\end{pmatrix}}}_M}
			\begin{pmatrix}
				1 & 0 \\
				0 & 1
			\end{pmatrix}
			\begin{pmatrix}
				1 & 0 \\
				0 & 2
			\end{pmatrix}
			\overbracket{
				\overbracket{\begin{pmatrix}
						0 & 1\\
						1 & 0
				\end{pmatrix}}^{P'}
				\begin{pmatrix}
					1 & -\frac12 & -\frac12 & \frac12 & 0 & 1 \\
					0 & 0 & 1 & -1 & 1 & 2
				\end{pmatrix}
			}^{P'U}
		\end{aligned}$
		\[c_1[0]=\begin{pmatrix}0&1\end{pmatrix}\qquad c_1[1]=\begin{pmatrix}2&0\end{pmatrix}\]
		\item $\begin{aligned}[t]
			&\operatorname{LDPU}\left[\begin{pmatrix}
				0 & 0 & 1 \\
				-1 & 1 & 2 \\
				1 & -\frac{1}{2} & -\frac{1}{2} \\
				\frac{1}{2} & 0 & 1
			\end{pmatrix}\right]\\
			&= \lefteqn{\underbracket{\phantom{
						\begin{pmatrix}
							1 & 0 & 0 \\
							0 & 1 & 0 \\
							0 & -1 & 1 \\
							\frac{1}{2} & -\frac{1}{2} & 1
						\end{pmatrix}
						\begin{pmatrix}
							1 & 0 & 0 \\
							0 & -1 & 0 \\
							0 & 0 & \frac{1}{2}
						\end{pmatrix}
						\begin{pmatrix}
							0 & 0 & 1 \\
							1 & 0 & 0 \\
							0 & 1 & 0
				\end{pmatrix}}}_M}
			\begin{pmatrix}
				1 & 0 & 0 \\
				0 & 1 & 0 \\
				0 & -1 & 1 \\
				\frac{1}{2} & -\frac{1}{2} & 1
			\end{pmatrix}
			\begin{pmatrix}
				1 & 0 & 0 \\
				0 & -1 & 0 \\
				0 & 0 & \frac{1}{2}
			\end{pmatrix}
			\overbracket{
				\overbracket{\begin{pmatrix}
						0 & 0 & 1 \\
						1 & 0 & 0 \\
						0 & 1 & 0
				\end{pmatrix}}^{P'}
				\begin{pmatrix}
					1 & -1 & -2 \\
					0 & 1 & 3 \\
					0 & 0 & 1
				\end{pmatrix}
			}^{P'U}
		\end{aligned}$
		\[c_2[0]=\begin{pmatrix}1&\frac12&0\\0&0&1\end{pmatrix}\qquad c_2[1]=\begin{pmatrix}\frac12&\frac12&\frac12\\-1&0&0\end{pmatrix}\]
		\item $\begin{aligned}[t]
			&\operatorname{LDPU}\left[\begin{pmatrix}
				0 & 0 & 1 & 1 & -1 & -2 & 0 & 1 & 3
			\end{pmatrix}^\intercal\right]\\
			&= \lefteqn{\underbracket{\phantom{
						\begin{pmatrix}
							0 & 0 & 1 & 1 & -1 & -2 & 0 & 1 & 3
						\end{pmatrix}^\intercal
						\begin{pmatrix}
							1
						\end{pmatrix}
						\begin{pmatrix}
							1
				\end{pmatrix}}}_M}
			\begin{pmatrix}
				0 & 0 & 1 & 1 & -1 & -2 & 0 & 1 & 3
			\end{pmatrix}^\intercal
			\begin{pmatrix}
				1
			\end{pmatrix}
			\overbracket{
				\overbracket{\begin{pmatrix}
						1
				\end{pmatrix}}^{P'}
				\begin{pmatrix}
					1
				\end{pmatrix}
			}^{P'U}
		\end{aligned}$
		\[c_3[0]=\begin{pmatrix}1\\0\\0\end{pmatrix}\qquad c_3[1]=\begin{pmatrix}-1\\1\\0\end{pmatrix}\qquad c_3[2]=\begin{pmatrix}-2\\3\\1\end{pmatrix}\]
	\end{enumerate}
	For a quick sanity check, one can check that we can indeed compute the last $2$ values of $T$, indexed respectively by $(1,1,1)$ and $(1,1,2)$, as:
	\[c_1[1]c_2[1]c_3[1]=\begin{pmatrix}0\end{pmatrix}\qquad
	\text{and}\qquad
	c_1[1]c_2[1]c_3[2]=\begin{pmatrix}2\end{pmatrix}\]
	
	\section{Proofs}
	
	\begin{proof}[Proof of \Cref{prop:LDPU-uniqueness}]
		The existence and complexity of the decomposition are given by \Cref{alg:LDPU}. Let's now prove the uniqueness. First, if $A=0$, the decomposition is obviously unique. Then, assume $A\neq 0$ has two LDPU decompositions: $A=L_1D_1P_1U_1=L_2D_2P_2U_2$, with the above constraints.
		
		Denote $(i_k)_{1\leq k\leq r}$ the pivot rows of $L_1$. By definition, for any $s>0$ and any $k$:
		\[i_k < i_{k+1},\quad L_1[i_k,k]\neq 0,\quad L_1[i_k,k+s] = 0,\quad L_1[i_k-s,k] = 0\]
		Notice that for a given $k$, $A[i_k]$ is linearly independent from the previous rows of $A$. The $i_k$s hence denote the indices of the rows of $A$ that are linearly independent from the previous ones. These pivot rows then have to be the same for $L_2$. Similarly, all column pivots of the $U_i$s have to be the same. Then, consider the first pivot row:
		\[A[i_1] = L_1[i_1,1]D_1[1,1]U_1[\sigma_1(1)] = L_2[i_1,1]D_2[1,1]U_2[\sigma_2(1)]\]
		where $\sigma_i$ is the inverse of the permutation induced by $P_i$. Since $L_1$ and $L_2$ are in unit echelon form, $L_1[i_1,1] = L_2[i_1,1] = 1$.
		Since $U_1$ and $U_2$ are full-rank, $U_1[\sigma_1(1)]$ and $U_2[\sigma_2(1)]$ are non-zero, hence they are colinear, so they must have the same leading index. Since they share the same pivot columns, we get $\sigma_1(1) = \sigma_2(1)$. Since both have $1$ as leading coefficient (since the $U_i$ are unit), $D_1[1,1] = D_2[1,1]$ and $U_1[\sigma_1(1)]=U_2[\sigma_2(1)]$.
		
		By induction, suppose that $\sigma_1(i)=\sigma_2(i)$,  $U_1[\sigma_1(i)]=U_2[\sigma_1(i)]$ and $D_1[i,i] = D_2[i,i]$ for all $j < k$. Then:
		\begin{align*}
		A[i_{k}]
		&= \sum_{\substack{0 < j< k\\\sigma_1(j)<\sigma_1(k)}} L_1[i_{k},j]D_1[j,j]U_1[\sigma_1(j)] + D_1[k,k]U_1[\sigma_1(k)]\\
		&= \sum_{\substack{0 < j< k\\\sigma_2(j)<\sigma_2(k)}} L_2[i_{k},j]D_2[j,j]U_2[\sigma_2(j)] + D_2[k,k]U_2[\sigma_2(k)]\\
		&= \sum_{\substack{0 < j< k\\\sigma_1(j)<\sigma_1(k)}} L_2[i_{k},j]D_1[j,j]U_1[\sigma_1(j)] + D_2[k,k]U_2[\sigma_2(k)]
		\end{align*}
		where we used the fact that $L_1[i_{k},k] = L_2[i_{k},k] = 1$.
		Since we only keep the terms where $\sigma_1(j)<\sigma_1(k)$, all the $U_1[\sigma_1(j)]$ have a smaller leading index than $U_i[\sigma_i(k)]$. By taking each $U_1[\sigma_1(j)]$ in order by ascending leading coefficient, we can see that it is forced that $L_1[i_{k},j] = L_2[i_{k},j]$. We are then left with $D_1[k,k]U_1[\sigma_1(k)]=D_2[k,k]U_2[\sigma_2(k)]$. With the same reasoning as for $k=1$ above, we get $\sigma_1(1) = \sigma_2(1)$, $D_1[1,1] = D_2[1,1]$ and $U_1[\sigma_1(1)]=U_2[\sigma_2(1)]$.
		
		At then end of this induction, we have shown that $D_1=D_2$, $\sigma_1=\sigma_2$, and $U_1 = U_2$.
		
		In the induction process, we have also shown that all pivot rows of $L_1$ and $L_2$ are equal. It is then easy to prove that non-pivot rows of $L_1$ and $L_2$ are also equal. Suppose $r$ is not a pivot of $L_1$ and $L_2$. If $r<i_1$, then $A[r]=0$, so $L_1[r]=L_2[r]=0$. Otherwise, let $i_k$ be the last pivot before $r$. We then have $A[r]=\sum\limits_{j=1}^k L_1[r,j]D_1[j,j]U_1[\sigma_1(j)]=\sum\limits_{j=1}^k L_2[r,j]D_2[j,j]U_2[\sigma_2(j)]=\sum\limits_{j=1}^k L_2[r,j]D_1[j,j]U_1[\sigma_1(j)]$. By linear independence of the family $\{U_1[\sigma_1(1)], ..., U_1[\sigma_1(k)]\}$, we get $L_1[r,j]=L_2[r,j]$ for all $j$, i.e.~that row $r$ of $L_1$ and $L_2$ are equal. This finishes the proof of uniqueness of the LDPU decomposition.
	\end{proof}
	
	\begin{proof}[Proof of \Cref{lem:REF-tensor-id}]
		The network above represents the Kronecker product of $U_0$ and $U_1$, which is in full row-rank (unit) row-echelon form if both $U_0$ and $U_1$ are.
	\end{proof}
	
	\begin{proof}[Proof of \Cref{lem:REF-compo}]
		A matrix in row-echelon form can be seen as being in the form $R=SU$ where $S$ is the identity matrix with missing rows, and $U$ is a square upper-triangular matrix. Then, the composition of two such matrices is $R_1R_2=S_1U_1S_2U_2$. Notice that the product $U_1S_2$ can be seen as inserting $0$-columns to $U_1$, resulting in a REF matrix $U_1S_2= S_3U_3$, so $R_1R_2=S_1S_3U_3U_2$. The product $S_1S_3$ can still be seen as the identity with missing rows, and $U_3U_2$ is upper-triangular, as the product of two upper-triangular matrices. Hence, the result is in REF. If $R_1$ and $R_2$ are both unit, we can pick $U_1$ and $U_2$ to be unit (a.k.a.~unitriangular). In which case we easily see that the result is unit.
	\end{proof}
	
	\begin{proof}[Proof of \Cref{lem:REF-right-inverse}]
		Let's prove both implications:
		\begin{itemize}
			\item[$\Leftarrow$)] A matrix in full row-rank row-echelon form can be seen as being in the form $R=SU$ where $S$ is the identity matrix with missing rows, and $U$ is a square invertible upper-triangular matrix. The condition on $S$ implies that $SS^\intercal=I$. We can then see that $R U^{-1}S^\intercal = I$, so $R$ has a right inverse $U^{-1}S^\intercal$ which is in REF by \Cref{lem:REF-compo} since both $U^{-1}$ and $S^\intercal$ are in REF. If $R$ is unit, we can choose $U$ to be unitriangular (i.e.~its diagonal entries are all $1$), in which case its inverse is also unitriangular. This makes the right inverse $U^{-1}S^\intercal$ a unit REF.
			\item[$\Rightarrow)$] Suppose $U$, with $n$ rows, has a right inverse $U^\diamond$: $UU^\diamond=I_n$. We then have $n = \operatorname{rk}(I_n) \leq \min(\operatorname{rk}(U),\operatorname{rk}(U^\diamond))\leq \operatorname{rk}(U)\leq n$. This forces $\operatorname{rk}(U)=n$ i.e.~$U$ full row-rank.\qedhere
		\end{itemize}
	\end{proof}
	
	\begin{proof}[Proof of \Cref{lem:REF-compo-frr}]
		By \Cref{lem:REF-compo}, the product is in REF. If the two matrices are full row-rank, then they have a right inverse, so the product also has a right inverse, which means it is also full row-rank by \Cref{lem:REF-right-inverse}.
	\end{proof}
	
	\begin{proof}[Proof of \Cref{lem:unique-REF-factorisation}]
		Since $U_1$ and $U_2$ are in full row-rank REF, they have a right inverse $U_1^\diamond$ and $U_2^\diamond$ by \Cref{lem:REF-right-inverse}, hence $XU_1 = X'U_1 = U_2$ implies $X=X'$. Moreover, since $U_1^\diamond$ is in REF, by \Cref{lem:REF-compo}, $X = U_2U_1^\diamond$ is also in REF. We also have $XU_1U_2^\diamond=I$, hence $X$ has a right inverse. Hence, $X$ must be full row-rank by \Cref{lem:REF-right-inverse}. If $U_1$ and $U_2$ are unit, so is $U_1^\diamond$, hence  $X = U_2U_1^\diamond$ is also unit by \Cref{lem:REF-compo}.
	\end{proof}
	
	\begin{proof}[Proof of \Cref{prop:algo-in-NF}]
	Let $(C_1,...,C_d)$ be a tensor train obtained from \Cref{alg:NF-from-full-tensor}, with associated permutation matrices $P_0,P_1,...,P_d$. The second condition of normal forms is easy to verify, for a given core $c$ is obtained as $\operatorname{untwine}(LDP',n,P)$, hence $\operatorname{intertwine}(c, n, P)P'^{-1} = LDP'P'^{-1}=LD$, which is in full column-rank column-echelon form.
	
	The first condition for the first core is obviously verified.
	Then, for $i>1$, we can see from the graphical interpretation that:
	\[\tikzfig{algo-partial}\]
	so that the cores $C_{i}$, ..., $C_d$ are built out of $U_{i-1}$, which is in full row-rank row-echelon form, for it is obtained by LDPU decomposition, i.e.:
	\def\fig{TT-trailing-U_k}
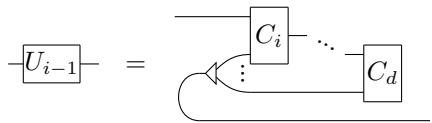
	\begin{align*}
		
		\eq{}
	\end{align*}
	We can then easily prove by induction, from $i=d$ down to $2$ that the first condition of normal forms is verified. For $i=d$, the result is obvious from the last equation,taking $i\leftarrow d$. Then, suppose the result has been proven for $i+1$, then:
	\def\fig{TT-trailing-U_k}
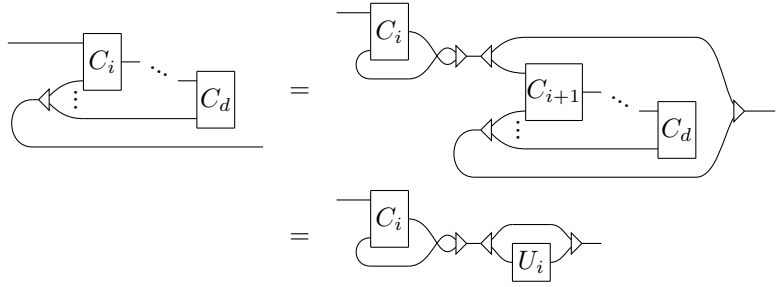
	\begin{align*}
		
		&\eq{}\input{./figures/\fig/\fig_02.tikz}\\
		&\eq{}\input{./figures/\fig/\fig_03.tikz}
	\end{align*}
	By \Cref{lem:REF-tensor-id} and by the induction hypothesis, the right part of the last diagram (the Kronecker product of the identity with $U_i$) is in full row-rank unit REF, and by the previous remarks, the whole diagram is itself in full row-rank unit REF. By \Cref{lem:unique-REF-factorisation}, \def\fig{1st-sweep}$\tikzfig{NF-REF-2}$ must be in full row-rank unit REF.
	\end{proof}
	
	\begin{proof}[Proof of \Cref{thm:NF-uniqueness}]
	The existence is provided by \Cref{alg:NF-from-full-tensor}, together with \Cref{prop:algo-in-NF}. Let's then prove the uniqueness.
	
	Let the tensor train $(C_1,...,C_d)$ represent tensor $T$ in normal form, with associated permutation matrices $P_0,P_1,...,P_d$. 
	We first show that the partial tensor train:
	\def\fig{TT-trailing-U_k}
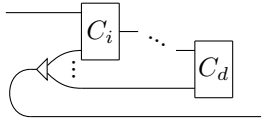
	\begin{align*}
		
	\end{align*}
	is full row-rank and in unit row-echelon form for $i>1$, by induction on $i$ form $d$ down to $2$. The case $i=d$ is provided by the first constraint on the normal form. Then, assuming the result is proven for $i+1$, we again have:
	\def\fig{TT-trailing-U_k}
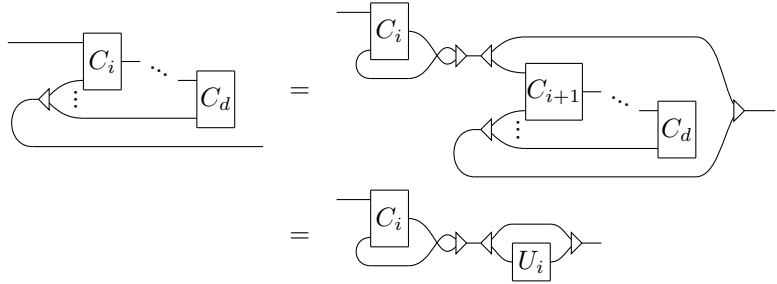
	\begin{align*}
		
		&\eq{}\input{./figures/\fig/\fig_02.tikz}\\
		&\eq{}\input{./figures/\fig/\fig_03.tikz}
	\end{align*}
	where $U_i$ is the partial tensor train starting at $i+1$, which by induction hypothesis is full row-rank and in unit REF. By the normal form hypothesis, the left part of the last diagram is also full row-rank and in unit REF, and by \Cref{lem:REF-tensor-id}, so is the right part. By \Cref{lem:REF-compo,lem:REF-compo-frr}, the resulting diagram is full row-rank and in unit REF.
	
	Similarly, we can show that the first part of the tensor train:
	\[\tikzfig{TT-partial-interp-CEF}\]
	is full column-rank and in column-echelon form for $1\leq i\leq d$, by induction on $i$. The case $i=1$ is given by the second constraint on the normal form. Then:
	\def\fig{TT-partial-interp-CEF-prf}
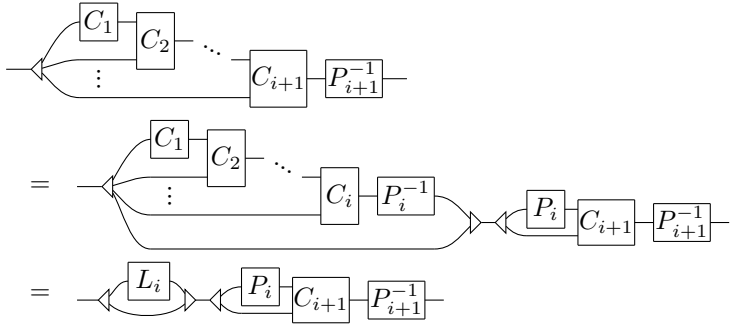
	\begin{align*}
		&\\
		&\eq{}\\
		&\eq{}\input{./figures/\fig/\fig_02.tikz}
	\end{align*}
	where $L_i$ is by the induction hypothesis in full column-rank CEF. Using the transpose of \Cref{lem:REF-tensor-id}, the left part of the last diagram is in full column-rank CEF. The right part is in full column-rank CEF by the normal form hypothesis, so, using the transpose of \Cref{lem:REF-compo,lem:REF-compo-frr}, the whole diagram is in full column-rank CEF.
	
	We can now show that each $C_i$ is uniquely defined, starting from $i=1$. By separating the first wire from the others, we get:
	\def\fig{NF-uniqueness-1}
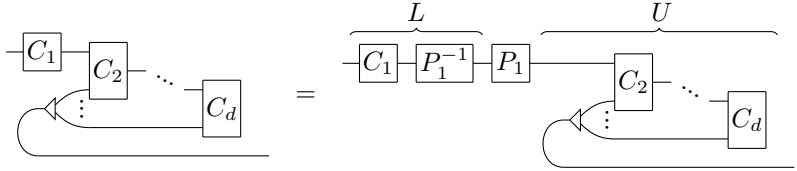
	\begin{align*}
		
		\eq{}
	\end{align*}
	where the left part $L$ is in full column-rank CEF and $P_1$-conditioned by hypothesis, and the right part $U$ is in full row-rank REF. We hence exactly get the LDPU decomposition of the tensor on the lhs (where $L$ contains the diagonal). This uniquely fixes $C_1$ and $P_1$, as well as the right part, hence together, it uniquely fixes
	\[\input{./figures/\fig/\fig_02.tikz}\]
	Suppose now by induction that all the $C_i$ and $P_i$ have been uniquely fixed for $i<k$, as well as:
	\def\fig{NF-uniqueness-k}
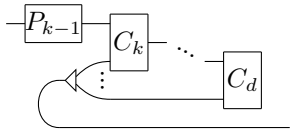
	\[\]
	By reshaping and introducing $P_k$, the following diagram is also uniquely defined:
	\[\]
	which is again, by the constraints of the normal form, in LDPU form. By uniqueness of the LDPU, $C_k$ and $P_k$ are uniquely fixed, as well as the composition of $C_{k+1}$, ..., $C_d$ on the right. We hence prove by induction that all the $C_k$ are uniquely fixed (as well as the permutations $P_k$). The normal form of tensor trains is unique.
	\end{proof}
	
	\begin{proof}[Proof of \Cref{prop:nullity-checking}]
		If $C_1$ is null, then the full tensor is also obviously null. We now prove the other implication. 
		Consider the following matrix obtained from the tensor train $C$, that we suppose to represent the null tensor:
		\[\tikzfig{nullity}\]
		Since the tensor train is obtained by the first sweep, the right part of the network (starting at $C_2$) is in full row-rank row-echelon form. Since the full tensor is null, we get that either $r_1=0$, or all the entries of $C_1$ are zero. In either case, $C_1$ is a null core.
	\end{proof}
	\begin{proof}[Proof of \Cref{prop:leading-index-coeff}]
	First, we can compute the complexity of the algorithm. For each $k$, looking for the leading index of $C_k[j,i]$ (and checking if it has one) can be done in $\mathcal O(n_kr_k)$. The other commands in the loop are in $\mathcal O(1)$, we hence have an algorithm that runs in $\mathcal O(\sum_{k=1}^d n_k r_k)$, which we can roughly simplify to $\mathcal O(dnr)$.
	
	Suppose that the leading index of tensor $T$ is $(\ell_1,...,\ell_d)$, then $T[\ell_1,...,\ell_d]\neq0$ and for any smaller index $(l_1,...,l_d)<(\ell_1,...,\ell_d)$ (in lexicographic order), $T[l_1,...,l_d]=0$. This also implies that for all $k$, if $l_k<\ell_k$, then $T[\ell_1,...,\ell_{k-1},l_k]=0$. We use this fact to find $\ell_k$ one index at a time, starting at $k=1$.
	
	Recall that:
	\[\tikzfig{li-lc-1}\]
	with $U_1$ in full row-rank REF. Because of this last property, $C_1[j]U_1$ is null iff $C_1[j]$ is null. We hence pick $\ell_1$ as the smallest index such that $C_1[\ell_1]\neq0$. We then store $i_1$ the leading index of $C_1[\ell_1]$.
	
	Suppose now we have the first $k$ indices $\ell_1,...,\ell_k$, and that the leading index of $v_k = C_1[\ell_1]...C_k[\ell_k]$ is $i_k$. Then:
	\def\fig{li-lc-k}
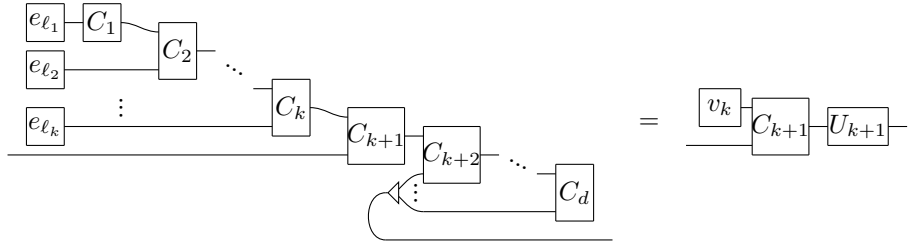
	\begin{align*}
		
		\eq{}
	\end{align*}
	where $U_{k+1}$ is in full row-rank REF, so looking for the smallest $j$ such that $v_kC_{k+1}[j]U_{k+1}\neq 0$ reduces to looking for the smallest $j$ such that $v_kC_{k+1}[j]\neq0$. Since $C_{k+1}[j]$ is itself in REF, and $v_k$ has leading index $i_k$, $v_kC_{k+1}[j]\neq0$ iff $C_{k+1}[j,i_k]\neq0$, looking for this $j$ then yields $\ell_{k+1}$. Once leading indices are found, computing the leading coefficient becomes easy, and can be done as we go.
	\end{proof}
	
	\begin{proof}[Proof of \Cref{thm:reduction}]
	Let's first compute the complexity of the algorithm, by assuming we apply it on tensor train $(C_1,...,C_d)$ with modes $n_1,...,n_d$ and TT-ranks $r_0=1,r_1,...,r_{d-1},r_d=1$. Let's write $s$ the size of this tensor train.
	Let's first compute the complexity of the first sweep. At core $k$, we compute the LDPU of $\begin{pmatrix}
		C_k[1]\mid \ldots
	\end{pmatrix}$ which is a $r_{k-1}\times n_kr_k$ matrix, hence in $\mathcal O(n_kr_kr_{k-1}\min(r_{k-1},n_kr_k))$. We also compute the product $C_{k-1}[i]LDP$ whose complexity is dominated by that of the product $C_{k-1}[i]L$ (since products by diagonal and permutation matrices are much simpler). This product is naively in $\mathcal O(r_{k-2}r_{k-1}\min(r_{k-1},n_kr_k))$, and is done $n_{k-1}$ times. The cost of building and breaking down the block matrices is dominated by these two costs. This is computed for $k\in\{2,...,d\}$, and on the one hand: \[\sum_{k=2}^d n_kr_kr_{k-1}\min(r_{k-1},n_kr_k) \leq r\sum_{k=2}^d n_kr_kr_{k-1}\leq rs\]
	On the other hand:
	\[\sum_{k=2}^d n_{k-1}r_{k-2}r_{k-1}\min(r_{k-1},n_kr_k)\leq r\sum_{k=2}^d n_{k-1}r_{k-2}r_{k-1}\leq rs\]
	Hence, the complexity of the 1st sweep lies in $\mathcal O(rs)$.
	
	For the second sweep, at core $k$, we do an LDPU of a $n_kr_{k-1}\times r_k$ matrix, hence in $\mathcal O(n_kr_kr_{k-1}\min(r_k,n_kr_{k-1}))$. We also compute the product $UC_{k+1}[i]$, which is naively in $\mathcal O(r_{k}r_{k+1}\min(r_{k-1}n_k,r_k))$. This product is done $n_{k+1}$ times. Again, the costs of $\operatorname{intertwine}$ and $\operatorname{untwine}$ are dominated by the above. These computations are done for all $k\in\{1,d-1\}$. On the one hand: \[\sum_{k=1}^{d-1}n_kr_kr_{k-1}\min(r_k,n_kr_{k-1})\leq r\sum_{k=1}^{d-1}n_kr_kr_{k-1}\leq rs\]
	On the other hand: \[\sum_{k=1}^{d-1}n_{k+1}r_{k}r_{k+1}\min(r_{k-1}n_k,r_k)\leq r\sum_{k=1}^{d-1}n_{k+1}r_{k}r_{k+1}\leq rs\]
	Both sweeps are done un $\mathcal O(rs)$, hence, so is the full algorithm.
	
	The first sweep obviously forces the tensor train to verify the first condition of the normal form. The second sweep also quite directly forces it to verify the second condition. One then simply has to check that this second sweep does not undo the first. Since we are only updating each core by 
	applying some $U$ to its first mode, and since that $U$ is obtained from an LDPU, and hence in full row-rank unit REF, by \Cref{lem:REF-compo,lem:REF-compo-frr}, the first condition on normal form is preserved.
	
	In both sweeps, we cannot augment the TT-ranks. They may however decrease, specifically when computing the LDPU decomposition of non-full-rank matrices. Since the TT-ranks cannot augment, the size of the tensor train cannot augment either.
	\end{proof}
	
	\begin{proof}[Proof of \Cref{prop:TT-upper-bound}]
	This is better seen from \Cref{alg:NF-from-full-tensor}, which produces a tensor train in normal form. 
	In the worst case, the LDPU is always applied on a full-rank matrix. In that case, the matrix in question, at step $k$, is of size $n_1...n_k\times n_{k+1}...n_d$, so the rank $r_k$ is $\min(n_1...n_k, n_{k+1}...n_d)$. Let $\ell$ be the "tipping point", i.e.~the index where:
	\[r_k=\begin{cases}
		\prod_{i=1}^k n_i& \text{ if } k\leq\ell\\
		\prod_{i=k+1}^d n_i& \text{ if } k>\ell
	\end{cases}\]
	Then:
	\begin{align*}
		s &= \sum_{k=1}^d n_kr_{k-1}r_k\\
		&= \sum_{k=1}^{\ell} n_k\prod_{i=1}^{k-1}n_i\prod_{i=1}^k n_i +
		n_{\ell+1}\prod_{i=1}^{\ell}n_i\prod_{i=\ell+2}^{d}n_i+\!\!
		\sum_{k=\ell+2}^{d}\!\! n_k\prod_{i=k}^{d}\!n_i\!\!\prod_{i=k+1}^d \!\!n_i\!\!\!\\
		&= \sum_{k=1}^{\ell} \prod_{i=1}^k n_i^2 +
		\prod_{i=1}^{d}n_i+
		\sum_{k=\ell+2}^{d} \prod_{i=k}^{d}n_i^2\\
		&= \prod_{i=1}^{\ell}n_i^2\sum_{k=1}^{\ell} \frac1{\prod_{i=k+1}^\ell n_i^2} +
		\prod_{i=1}^{d}n_i+
		\prod_{i=\ell+2}^{d}n_i^2\sum_{k=\ell+2}^{d} \frac1{\prod_{i=\ell+2}^{k-1}n_i^2}\\
		&\leq \prod_{i=1}^{\ell}n_i^2\sum_{k=0}^{\ell-1} \frac1{4^k} +
		\prod_{i=1}^{d}n_i+
		\prod_{i=\ell+2}^{d}n_i^2\sum_{k=0}^{d-\ell-2} \frac1{4^k}\\
		&\leq \frac43 \prod_{i=1}^{\ell}n_i^2 +
		\prod_{i=1}^{d}n_i+
		\frac43\frac1{n_{\ell+1}^2}\prod_{i=\ell+1}^{d}n_i^2\\
		&\leq \left(\frac43+1+\frac13\right)\prod_{i=1}^{d}n_i
		=\frac83\prod_{i=1}^{d}n_i
	\end{align*}
	where we used $n_i\geq2$, $\sum_{k=0}^*\frac1{4^k}\leq \frac43$, as well as $\prod_{i=1}^\ell n_i^2\leq \prod_{i=1}^d n_i$ and $\prod_{i=\ell+1}^{d}n_i^2\leq \prod_{i=1}^d n_i$.
	
	Let's now show that this bound is reached up to an additive constant when $n_1=...=n_d=2$. Assume $d$ is even: $d=2d'$. Then $\prod_{i=1}^dn_i = 4^{d'}$. We can see that $r_k = 2^k$ and $r_{d'+k}=2^{d'-k}$ for $k\leq d'$. Hence:
	\begin{align*}
		s&=\sum_{k=1}^d n_kr_{k-1}r_k
		= \sum_{k=1}^{d'} n_kr_{k-1}r_k + \sum_{k=d'}^{2d'-1} n_kr_{k-1}r_k
		= 2\sum_{k=1}^{d'}r_{k-1}r_k + 2\sum_{k=0}^{d'-1}r_{k+d'}r_{k+d'+1}\\
		&= 2\sum_{k=0}^{d'-1}2^k\cdot2^{k+1} + 2\sum_{k=0}^{d'-1}2^{d'-k}2^{d'-k-1}
		= 2\sum_{k=0}^{d'-1}2^k\cdot2^{k+1} + 2\sum_{k=0}^{d'-1}2^k\cdot2^{k+1}\\
		&=8\sum_{k=0}^{d'-1}4^k
		=8\frac{4^{d'}-1}3 = \frac83\prod_{i=1}^dn_i - \frac83 \qedhere
	\end{align*}
	\end{proof}
	
\end{document}